\documentclass[11pt]{article}

\usepackage{natbib} 
\usepackage{JASA_manu}

\usepackage{amsmath,amssymb,amsfonts, amsbsy, amsthm}
\usepackage[usenames]{color}
\usepackage{rotating}
\usepackage{multirow}
\usepackage{graphics,epsfig,subfig}
\usepackage{latexsym}
\usepackage{indentfirst}
\usepackage{ulem}
\usepackage{booktabs}

\newtheorem{theorem}{Theorem}
\newtheorem{corollary}{Corollary}
\newtheorem{proposition}{Proposition}
\newtheorem{lemma}{Lemma}

\newcommand{\propref}[1]{Proposition~\ref{#1}}

\newcommand{\lemref}[1]{Lemma~\ref{#1}}

\newcommand{\fwer}{\textsc{fwer}}
\newcommand{\fdr}{\textsc{fdr}}
\newcommand{\nettpr}{\textsc{nettpr}}

\newcommand{\Card}{\textsf{Card}}
\newcommand{\Cov}{\textsf{Cov}}
\newcommand{\Cor}{\textsf{Cor}}
\newcommand{\diag}{\textsf{diag}}

\newcommand{\Var}{\textsf{Var}}
\renewcommand{\P}{\textsf{P}}
\newcommand{\E}{\textsf{E}}

\newcommand{\trans}{\mathrm{\scriptscriptstyle T}}

\newcommand{\normal}{\mathrm{N}}

\newcommand{\etc}{\textit{etc}.}
\newcommand{\ie}{\textit{i.e.}}

\newcommand{\ud}{\,\mathrm{d}}
\newcommand{\uH}{\,\mathrm{H}}
\newcommand{\uU}{\,\mathrm{U}}

\newcommand{\RR}{\mathbb{R}}

\newcommand{\Ar}{\mathcal{A}}

\newcommand{\Dr}{\mathcal{D}}

\newcommand{\Fr}{\mathcal{F}}

\newcommand{\Ir}{\mathcal{I}}
\newcommand{\Mr}{\mathcal{M}}
\newcommand{\Or}{\mathcal{O}}
\newcommand{\Qr}{\mathcal{Q}}

\newcommand{\Ur}{\mathcal{U}}
\newcommand{\Vr}{\mathcal{V}}
\newcommand{\Sr}{\mathcal{S}}
\newcommand{\Tr}{\mathcal{T}}

\newcommand{\va}{\mathbf{a}}
\newcommand{\vA}{\mathbf{A}}
\newcommand{\vb}{\mathbf{b}}
\newcommand{\vB}{\mathbf{B}}

\newcommand{\vD}{\mathbf{D}}

\newcommand{\vI}{\mathbf{I}}

\newcommand{\vN}{\mathbf{N}}

\newcommand{\vR}{\mathbf{R}}

\newcommand{\vU}{\mathbf{U}}
\newcommand{\vV}{\mathbf{V}}
\newcommand{\vW}{\mathbf{W}}
\newcommand{\vx}{\mathbf{x}}

\newcommand{\bzero}{\boldsymbol{0}}

\newcommand{\bN}{\boldsymbol{N}}
\newcommand{\bX}{\boldsymbol{X}}

\newcommand{\by}{\boldsymbol{y}}

\newcommand{\eps}{\epsilon}
\newcommand{\vareps}{\varepsilon}

\newcommand{\balpha}{\boldsymbol{\alpha}}
\newcommand{\bbeta}{\boldsymbol{\beta}}
\newcommand{\bdel}{\boldsymbol{\delta}}

\newcommand{\bmu}{\boldsymbol{\mu}}

\newcommand{\bLam}{\boldsymbol{\Lambda}}
\newcommand{\bSigma}{\boldsymbol{\Sigma}}
\newcommand{\bOmega}{\boldsymbol{\Omega}}
\newcommand{\bUps}{\boldsymbol{\Upsilon}}

\newcommand{\bvareps}{\boldsymbol{\varepsilon}}

\def\rT{\scriptsize{\mathrm T}}

\newcommand\independent{\protect\mathpalette{\protect\independenT}{\perp}}
\def\independenT#1#2{\mathrel{\rlap{$#1#2$}\mkern2mu{#1#2}}}

\def\text#1{\mbox{\rm #1}}

\newcommand{\cm}[1]{\ignorespaces}

\newcommand{\ignore}[1]{}

\begin{document}

\title{High Dimensional Tests for Functional Networks of Brain
  Anatomic Regions} \author{Jichun Xie and Jian Kang}

\maketitle

\mbox{} 
\vspace*{1in} 
\begin{center} 
\textbf{Author's Footnote:} 
\end{center} 
Jichun Xie is Assistant Professor, Department of Biostatistics and
Bioinformatics, Duke University School of Medicine, Durham, NC 27705. (Email:
jichun.xie@duke.edu). Jian Kang is
Assistant Professor, Department of Biostatistics, University of Michigan, Ann Arbor, MI 48109. (Email: jiankang@umich.edu).  Jian Kang's
research was partially supported by the NIH grant R01 MH105561. The authors thank the autism brain imaging data
exchange (ABIDE) study \citep{di2013autism} shares the resting-state
fMRI data.

\newpage
\begin{abstract}
  There has been increasing interests in learning resting-state brain
  functional connectivity of autism disorders using functional
  magnetic resonance imaging (fMRI) data. The data in a standard brain
  template consist of over 200,000 voxel specific time series for each
  single subject. Such an ultra-high dimensionality of data makes the
  voxel-level functional connectivity analysis (involving four billion
  voxel pairs) lack of power and extremely inefficient. In this work,
  we introduce a new framework to identify functional brain network at
  brain anatomic region-level for each individual. We propose two
  pairwise tests to detect region dependence, and one multiple testing
  procedure to identify global structures of the network. The limiting
  null distributions of the test statistics are derived. It is also
  shown that the tests are rate optimal when the alternative networks
  are sparse. The numerical studies show the proposed tests are valid
  and powerful. We apply our method to a resting-state fMRI study on
  autism and identify patient-unique and control-unique hub
  regions. These findings are consistent with autism clinical
  symptoms.
 \end{abstract}

\vspace*{.3in}


\section{Introduction}
The functional brain network refers to the coherence of the brain
activities among multiple spatially distinct brain
regions. It plays an important role in information processing and
mental representations~\citep{bullmore2009complex,
  sporns2004organization}, and could be altered by one's disease
status. \citet{supekar2008network,
  koshino2005functional,cherkassky2006functional} showed that patients
with neurodegenerative diseases (such as the Alzheimer's disease and
the Autism Spectrum Disorder) have different function network compared with
controls. As a result, the inference on functional brain network will
benefit the study of these diseases. Our research goal is to infer the
whole functional networks of the brain regions.

Recent advances in the neuroimaging technologies provide great
opportunities for researchers to study functional brain network based
on massive nueroimaging data, which are generated using various
imaging modalities such as positron emission tomography (PET),
functional magnetic resonance imaging (fMRI), and
electroencephalography (EEG). In a neuroimaging experiment, the
scanner records the brain signals over multiple times at each location
(or voxel) in the three-dimensional brain, leading to a
four-dimensional imaging data structure. In a typical fMRI study, the
number of voxels can be up to 200,000 and the number of imaging scans
over time is round 100--200. In light of the brain function and the
neuroanatomy, the human brain can be partitioned to 100-200 anatomical
regions and each region contains 200 to 4,000 voxels. Such high
dimensionality and complexity of the data imposes great challenges on
the inference of the whole brain network.

Due to the ultra-high dimensionality of voxel numbers (up to 200,000),
direct inference on the network of voxels is extremely computationally
expensive. More importantly, the network of interest is the network of
brain regions, not voxels. To this end, \citet{andrews2007disruption}
examines the functional connectivity of a particular brain region,
called seed region, by correlating the seed region brain signals
against the brain signals from all other regions. Although this method
yields a clear view of the functional connectivities between one
region of interest (the seed region) and other regions
\citep{biswal1995functional, cordes2000mapping}, it fails to examine
the functional network on a whole brain scale. Alternatively,
\citet{VeliogluFunctional2014} proposed to form meshes around a seed
voxel by regressing $p$ functionally nearest neighbor voxels on the
seed voxel, where number of regressors $p$ is determined by minimizing
the Akaike's final prediction error \citep{AkaikeFitting1969}. Then two
voxels are considered as functionally connected if one serves as a
functional predictor as the other. The number of all connected voxel
pairs between two anatomic regions are treated as the dependence level
between these two regions.  Although this method successfully provides
a functional network among anatomic regions, no inference results are
provided on what level of connectivities should be regarded as
significant. Another commonly used method~\citep{huang2009learning,
  HuangLearning2010} is to summarize one statistic (such as the
largest principal component of voxel signals) in each region and then
study the dependence between these statistics. Commonly used measures
of dependence include covariance matrix or Gaussian Graphical
model. See \citet{supekar2008network, weiss2001correctness,
  huang2009learning, marrelec2006partial}. Since only one statistic is
summarized in each region, the dependence among these summarized
statistics sometimes fail to represent the dependence among the
regions.

In this article, we propose a new method to estimate the region-level
functional connectivity for each individual. Instead of summarizing
one statistic in each region, we summarize multiple statistics so that
information of the region can be adequately captured. These statistics
can be viewed as functional components of the region. The correlation
matrix between the components in two regions are used to measure the
dependence between two regions. We assume that two regions are
functionally connected if and only if at least one pair of components
are correlated between these two regions. 

We then concatenate these functional components region by region. No
region-level functional connectivity implies that the covariance
matrix (or equivalently its inverse) of the concatenated components
has a block-diagonal structure. This is a reasonable assumption and
has been used in many existing literatures. (See
\citet{rubinov2010complex, bowman2012determining, huang2009learning}.)
Thus, to construct a functional network of brain anatomic regions, we
check if the correlation matrix of two regions has a block diagonal
structure.

Previous literatures for testing high dimensional
covariance/correlation matrix include testing whether the covariance
matrix is proportional to the identity matrix \citep{LedoitSome2002,
  BirkeNote2005,SchottTest2007,ChenTests2010,CaiOptimal2013,LiHypothesis2014},
and testing whether two covariance matrices are equal
\citep{LiTwo2012, CaiTwo2013, LiHypothesis2014}. To the best of our
knowledge, no existing methods have been proposed to address whether a
rectangle block of a covariance matrix is zero. However, ideas in
those literatures can be borrowed to construct test statistics for our
problem. There are mainly two types of existing test statistics: one
is chi-square type of statistic based on the sum square of sample
covariances.  and the other is the extreme type of statistic based on
the largest absolute self-standardized sample covariance. In general,
the chi-square type of statistics performs better when the alternative
network is dense and the extreme type of statistics performs better
when the alternative network is sparse. In imaging studies, the
network of functional components is usually sparse.  Therefore, we
will use the extreme type of statistics. Details will be discussed in
Section~\ref{sec:test}.

The rest of the paper is organized as follows. In
Section~\ref{sec:hypothesis}, we introduce the notations and define
the testing hypotheses of our interests. Section~\ref{sec:test}
presents two procedures to control type I error of each hypothesis and
a multiple testing procedure to control family-wise error
rate. Theoretical properties of the proposed procedures are discussed
in Section~\ref{sec:theory}, and their numerical performances are
shown in Section~\ref{sec:simul}. We apply the proposed procedures on
a resting-state fMRI data of subjects with and without autism spectrum
disorder (ASD), and compare the functional networks of anatomic
regions between cases and controls. The results match the clinical
characteristics of ASD.

\section{Model and Hypotheses}\label{sec:hypothesis}

In fMRI studies, blood-oxygen-level dependent (BOLD) signals are
collected at a large number of voxel locations for $n$ scans.  The
standard preprocessing steps including motion correction, slice-timing
correction, normalization, de-trending and de-meaning procedures are
applied to the BOLD signals \citep{worsley2002general,
  friman2005resampling, LindquistStatistical2008}, and then the
signals are clustered based on their voxel locations mapping to the
existing anatomic regions. After clustering, the signals are
summarized into functional components to reduce the dimension of
voxels and eliminate the redundancy of high coherent signals.  One way
to summarize the functional components is to perform principal
component analysis (PCA) in region $s$ to extract the first $q_s$
principal components. Alternatively, independent component analysis
(ICA) can be perform to extract $q_s$ independent components. The
choice of summarizing method depends on the distribution of the
processed signals. See
\citet{AndersonIntroduction2003,SamworthIndependent2012}. 

For each patient, assume that $q_s$ functional components are
summarized in region $s$. Each functional component is of length $n$,
containing replications of signals across $n$ scans. After removing
the temporal-correlation between the scans, denote by $X_{k,s,i}$ the
$k$-th scan of the $i$-th component in $s$-th brain region. 
Then these components can be treated as independent across scans.

Denote by $\bX_{k,s}= (X_{k,s,1},\ldots,X_{k,s,q_s})^\trans$ the
vector of functional components in region $s$ of scan $k$, and by
\[\bUps_{st} = \Cor(\bX_{k,s},\bX_{k,t})\]
the correlation matrix between region $s$ and region $t$. To test
whether region $s$ and region $t$ are functionally connected, we set
up the hypotheses:
\begin{equation}\label{eq:hypo1}
  \uH_{0,st}: \bUps_{st} = \bzero, \qquad \mbox{versus}\qquad 
  \uH_{1,st}: \bUps_{st} \neq \bzero.
\end{equation}
A rejection of $H_{0,st}$ implies that regions $s$ and region $t$ have
significant functional connectivity. The goal is to test $\uH_{0,st}$
with controlled type I error, and also to perform multiple testing on
$\uH_{0,st}$ simultaneously to control family-wise error rate.

The difficulty of this testing problem lies in the large number of
parameters and relatively small number of replications. First, the 
number of summarized functional components in each region may increase
with the number of scans $n$. Second, the number of
total region pairs $p(p-1)/2$ usually largely exceeds $n$. Therefore,
we need to address the high dimensional challenges in testing each
hypothesis and testing a large number of them simultaneously.

\section{Testing Procedures}\label{sec:test}

To test $\uH_{0,st}$, we propose two testing procedures to fit
different distribution assumptions of the functional
components. Therefore, neither of them can universally outperform the
other. We further develop a multiple testing procedure to control
the family-wise error (FWER) for testing $\{\uH_{0,st}:\ 1\leq s<t
\leq p\}$ simultaneously.

\subsection{Test I: Marginal Dependence Testing}

The first procedure is based on the Pearson correlation between the
components in two regions. 

Denote by the pairwise correlation $\rho_{st,ij} =
\Cor(X_{k,s,i},X_{k,t,j})$. Then the null hypothesis $\uH_{st,0}:
\bUps_{st} = \bzero$ is equivalent to $\uH_{st,0}: \max_{1\leq i\leq
  q_s, 1\leq j\leq q_t}|\rho_{st,ij}| = 0$. A straightforward approach
is to check whether the sample correlation between two regions is
close to zero. Denote the Pearson correlation
between the $i$-th component in region $s$ and the $j$-th component in
region $t$ by $\hat{\rho}_{st,ij}$, \ie,
\[\hat{\rho}_{st,ij} =
\hat{\sigma}_{st,ij}/\left(\hat{\sigma}_{ss,ii}\hat{\sigma}_{tt,jj}\right)^{1/2},\]
where $\bar{X}_{s,i} = \sum_{k=1}^n X_{k,s,i}/n$, $\bar{X}_{t,j} =
\sum_{k=1}^n X_{k,t,j}/n$, $\hat{\sigma}_{st,ij} =
\frac{1}{n}\sum_{k=1}^n(X_{k,s,i}-\bar{X}_{s,i})(X_{k,t,j}-\bar{X}_{t,j})$
is the sample covariance between the $i$-th component in region $s$
and the $j$-th component in region $t$, and $\hat{\sigma}_{ss,ii}$ and
$\hat{\sigma}_{tt,jj}$ are sample variances defining in the similar
manner. The test statistic is defined as
\begin{equation}
  \label{eq:Tst1}
  T_{st}^{(1)} = n\cdot \max_{i,j}\hat{\rho}_{st,ij}^2 -2\log (q_sq_t)
+ \log\log (q_sq_t). 
\end{equation}
With mild conditions (details in
Section~\ref{sec:theory}), under $\uH_{0,st}$, $T_{st}^{(1)}$
asymptotically follows the Gumbel distribution
\begin{equation}
F(x) = \exp\{-\pi^{1/2}\exp(-x/2)\}.\label{eq:null_dist}
\end{equation}
To control type I error at level $\alpha$, we reject $\uH_{0,st}$
if $T_{st}^{(1)}$ exceeds the $(1-\alpha)$-th quantile of $F(x)$, \ie,
$T_{st}^{(1)} > q_\alpha$, with
\begin{equation}\label{eq:qalpha}
q_\alpha= -\log(\pi) - 2\log\log\{1/(1-\alpha)\}.
\end{equation}

\subsection{Test II: Local Conditional Dependence Testing}

The alternative testing procedure is based on the Pearson correlation
between the residuals of local neighborhood selection in two regions. 

In region
$s$, we regress on each component $X_{k,s,i}$ the rest of components,
\begin{equation}
  X_{k,s,i}   = \alpha_{s,i} + \bX_{k,s,-i}^\trans\bbeta_{s,i} +
  \vareps_{k,s,i},\label{eq:ns1}
\end{equation}
where $\bX_{k,s,-i}$ is the vector of $\bX_{k,s}$ by removing the
$i$-th component. In region $t$ with $t\neq s$, we build up similar
regression model 
\begin{equation}
  X_{k,t,j}   = \alpha_{t,j} + \bX_{k,t,-j}^\trans\bbeta_{t,j} +
  \vareps_{k,t,l},\label{eq:ns2}
\end{equation}
Let $\rho_{\vareps,st,ij} =
\Cor(\vareps_{k,s,i},\vareps_{k,t,j})$ be the correlation of the error
terms in two models. Clearly, the null hypothesis $\uH_{0,st}$ is equivalent to
\[\uH_{0,st}:\ \max_{i,j}\rho_{\vareps,st,ij} =0.\]

We therefore develop a testing procedure to test if the correlations
$\rho_{\vareps,st,ij}$ are all zero. If the coefficients
$\bbeta_{s,i}$ and $\bbeta_{t,j}$ in model (\ref{eq:ns1}) and
(\ref{eq:ns2}) were known, we would know the value of each realization
of the random error $\vareps_{k,s,i}$ and $\vareps_{k,t,j}$, and
center them as $\tilde{\vareps}_{k,v,l} = \vareps_{k,v,l} -
\bar{\vareps}_{v,l}$ with $\bar{\vareps}_{v,l} =
\frac{1}{n}\sum_{k=1}^n \vareps_{k,v,l}$, $(v,l) = (s,i)$ or
$(v,l) = (t,j)$. Based on model (\ref{eq:ns1}) and (\ref{eq:ns2}), the
centered realization of randome error $\tilde{\vareps}_{k,v,l}$ could be
expressed as 
\begin{equation}
  \tilde{\vareps}_{k,v,l}  = X_{k,v,l}-\bar{X}_{v,l} 
  (\bX_{k,v,-l}-\bar{\bX}_{v,-l})^\trans\bbeta_{v,l}, \quad
  (v,l)=(s,i) \text{ or } (v,l) = (t,j).
  \label{eq:tdvareps}
\end{equation}
Consequently, the Pearson correlation between
$\tilde{\vareps}_{k,s,i}$ and $\tilde{\vareps}_{k,t,j}$ would be
\[\tilde{\rho}_{\vareps,st,ij} = \frac{1}{n}\sum_{k=1}^n
\tilde{\sigma}_{\vareps,st,ij}/\left(\tilde{\sigma}_{\vareps,ss,ii}
  \tilde{\sigma}_{tt,jj}\right)^{1/2},\] 
where $
\tilde{\sigma}_{\vareps,st,ij} =
\frac{1}{n}\sum_{k=1}^n\tilde{\vareps}_{k,s,i}\tilde{\vareps}_{k,t,j}$,
$ \tilde{\sigma}_{\vareps,ss,ii} =
\frac{1}{n}\sum_{k=1}^n\tilde{\vareps}_{k,s,i}^2$, and
$\tilde{\sigma}_{\vareps,tt,jj} =
\frac{1}{n}\sum_{k=1}^n\tilde{\vareps}_{k,t,j}^2 $.

Unfortunately in practice, the coefficients in (\ref{eq:ns1}) and
(\ref{eq:ns2}) are unknown. However, the coefficients can be well
estimated by existing methods, such as Lasso or Dantzig selector.
Suppose ``good''\footnote{We
  will discuss the criteria of ``good'' and how to obtain ``good''
  coefficient estimators in Section~\ref{sec:theory}.}
coefficient estimators $\hat{\bbeta}_{s,i}$ and $\hat{\bbeta}_{t,j}$
exist. Then the centered error term $\tilde{\vareps}_{k,v,l}$ can be estimated by 
\begin{equation}
\hat{\vareps}_{k,v,l} = X_{k,v,l}-\bar{X}_{v,l} -
(\bX_{k,v,-l}-\bar{\bX}_{v,-l})^\trans\hat{\bbeta}_{v,l},\quad
(v,l)=(s,i)\text{ or }(v,l)=(t,j).\label{eq:htvareps}
\end{equation}
Consequently, we calculate Pearson correlation based on
$\hat{\vareps}_{k,s,i}$ and $\hat{\vareps}_{k,t,j}$,
\[\hat{\rho}_{\vareps,st,ij} =
\hat{\sigma}_{\vareps,st,ij}/\left(\hat{\sigma}_{\vareps,ss,ii}
  \hat{\sigma}_{tt,jj}\right)^{1/2},\]
where 
 $\hat{\sigma}_{\vareps,st,ij} =
\frac{1}{n}\sum_{k=1}^n\hat{\vareps}_{k,s,i}\hat{\vareps}_{k,t,j}$,
$ \hat{\sigma}_{\vareps,ss,ii} =
\frac{1}{n}\sum_{k=1}^n\hat{\vareps}_{k,s,i}^2$, and
$\hat{\sigma}_{\vareps,tt,jj} =
\frac{1}{n}\sum_{k=1}^n\hat{\vareps}_{k,t,j}^2$.

Similar as Test I, we obtain
the test-statistics as follows.
\[T^{(2)}_{st} = n \cdot \max_{i,j}\hat{\rho}_{\vareps,st,ij}^2 -
2\log(q_sq_t) + \log\log (q_sq_t).\] Under certain condtions (discussed
in Section~\ref{sec:theory}) and $\uH_{0,st}$, $T_{st}^{(2)}$ also
follows the distribution $F(x)$ in (\ref{eq:null_dist}).  Therefore,
to control type I error at level $\alpha$, we reject $\uH_{0,st}$ if
$ T^{(2)}_{st}>q_\alpha$, where $q_\alpha$ is the $(1-\alpha)$-th quantile of
$F(x)$.

\subsection{Family-Wise Error Rate Control}

Considering the standard space of the brain \citep[Montreal
Neurological Institute, MNI]{mazziotta1995probabilistic} and the
commonly used brain atlas: the Automated Anatomical Labeling
\citep[AAL]{tzourio2002automated} regions, the number of region pairs
in the whole brain is over 4,000, which is much larger than the number
of scans (typically a couple of hundreds). This motivates the needs of
correction for multiplicity when testing any two of them are
connected, in order to detect the functional connectivity of the whole
brain. We propose procedure (\ref{eq:fwer}) to test $\{\uH_{0,st}: \
1\leq s<t\leq p\}$ simultaneously and control the family-wise error
rate (\fwer). The procedure can involve either $\widetilde T_{st}^{(1)}$ or 
$\widetilde T_{st}^{(2)}$, depending on the structure assumption of the
dependence structure of local voxels. It turns out that to control
\fwer\ at level $\alpha$, we only need to adopt a higher
threshold. The adjusted testing procedure is as follows:
\begin{equation}
  \text{Reject } \uH_{0,st} \text{ if and only if } T_{st}^{(b)} >
  2\log \{p(p-1)/2\} + q_\alpha \quad (1\leq s<t\leq p), \label{eq:fwer}
\end{equation}
for $b = 1,2$. The threshold depends on the desired family-wise
error rate $\alpha$, and the total number of region pairs $p(p-1)/2$.

\section{Theory}\label{sec:theory}
In this section, we show the null distributions of the test
statistics in procedures I and II, their power, and the optimality
properties of the proposed tests. Also, we prove that the
multiple testing procedure (\ref{eq:fwer}) is able to control
family-wise error rate.

For the rest of the paper, unless otherwise stated, we use the
following notations. For a vector $\va= (a_1,\ldots,a_p)^\trans \in
\RR^p$, denote by $\lvert \va \rvert_2 = (\sum_{j=1}^p a_j^2)^{1/2}$
its Euclidean norm. For a matrix $\vA = (a_{ij}) \in \RR^{p\times q}$,
define the spectral norm $\lVert \vA \rVert_2 = \sum_{\lvert
  \vx\rvert_2=1} \lvert \vA\vx\rvert_2$ and the Frobenius norm $\lVert
\vA\rVert_F = (\sum_{ij} a_{ij}^2)^{1/2}$. For a finite set
$\Ar=\{a_1,\ldots,a_s\}$, $\Card(\Ar)=s$ counts the number of elements
in $\Ar$. For two real number sequences $\{a_n\}$ and $\{b_n\}$, write
$a_n=O(b_n)$ if $\lvert a_n\rvert \leq C\lvert b_n \rvert$ hold for a
certain positive constant $C$ when $n$ is sufficiently large; write
$a_n=o(b_n)$ if $\lim_{n\rightarrow\infty} a_n/b_n=0$; and write $a_n
\asymp b_n$ if $c\lvert b_n \rvert \leq \lvert a_n\rvert \leq C\lvert
b_n\rvert$, for some positive constants $c$ and $C$ when $n$ is
sufficiently large.

Also assume the number of variables in all regions are comparable, \ie,
$q_1 \asymp q_2 \ldots \asymp q_p$. Let $q_0 =
\max(q_1,\ldots,q_p)$. Assume $\bX_{1,v},\ldots,\bX_{n,v}$ are
independently and identically distributed for each region $v$.

\subsection{Asymptotic Properties for Test I}\label{sec:theory_p1}

Denote by $\Upsilon_{vv} = (\rho_{vv,ij})_{q_v\times q_v}$ the
correlation matrix of $\bX_{k,v}$. For
$X_{k,v,i}$, denote by $r^{(1)}_{v,i}$ the number of other components
in region $v$ that at non-negligibly correlated with $X_{k,v,i}$,
\begin{align*}
  r^{(1)}_{v,i} & = \Card\{j: \lvert \rho_{vv,ij} \rvert \geq (\log
  q_0)^{-1-\alpha_0},\ j\neq i \},
\end{align*}
where $\alpha_0$ is a positive constant.
For a positive constant $\rho_0<1$, define
\[
  \Dr^{(1)}_v = \{i: \lvert \rho_{vv,ij} \rvert>\rho_0 \text{ for some }
  j \neq i\},
\]
Thus, $\Dr^{(1)}_v$ contains index $i$ such that
$X_{k,v,i}$ is highly correlated to at least one other
component in region $v$. 

We need the following conditions:

\noindent {\bf (C1.1)} For region $v=s,t$, there exists a subset $\Mr_v \subset
\{1,\ldots,q_v\}$ with $\Card(\Mr_v) = o(q_v)$ and a constant
$\alpha_0>0$ such that for all $\gamma>0$, $\max_{
  i\in \Mr_v^c} r^{(1)}_{v,i} = o(q_v^\gamma)$.  Moreover,
assume there exists a constant $0\leq \rho_0<1$ such that
$\Card\{\Dr^{(1)}_v\} = o(q_v)$. 

Condition (C1.1) constraints the sparsity level of non-neglegible and
large signals. It specifies that for each region $v$, for almost all
component $i$ within the region, the count of non-neglible
$\lvert\rho_{vv,ij}\rvert$ is of a smaller order of $q_v^\gamma$.
The condition is weaker than the commonly seen condition which
imposes a constant upper bound on the largest eigenvalue of
$\bSigma_{vv}$. In fact, if $\lambda_{\max}(\bSigma_{vv})= o\{q_v^\gamma/(\log
q_0)^{1+\alpha_0}\}$, $\max_{1\leq i\leq q_v} r^{(1)}_{v,i} =
o(q_v^\gamma)$.  In addition, (C1.1) also requires the number of
components that are very highly correlated with at least one other component
to be small. This condition can be easily satisfied if all the correlations
$\rho_{vv,ij}$ are bounded by $\rho_0$.

\noindent {\bf (C1.2)} Sub-Gaussian type tails: For region $v=s,t$, suppose that $\log(q_v) =
o(n^{1/5})$. There exist some constants $\eta >0$ and $K>0$ such that
\[
\max_{1\leq i\leq q_v}\E \left[\exp\{\eta (X_{k,v,i} - \mu_{v,i})^2/\sigma_{vv,ii}\}
\right] \leq K.\]

\noindent {\bf (C1.2*)} Polynomial-type tails: For region $v=s,t$, suppose that for some $\gamma_1$,
$c_1>0$, $q_0 \leq c_1 n^{\gamma_1+1/2}$, and for some $\epsilon>0$,
\[\max_{1\leq i\leq q_v}\E\lvert (X_{k,v,i}-\mu_{xi})/\sigma_{vv,ii}^{1/2} \rvert^{4\gamma_1 + 4 +
  \epsilon} \leq K. \] 

Conditions (C1.2) and C(1.2*) impose constraints on the tail of the
distribution of $X_{k,v,i}$, and the corresponding order of
$q_v$. They fit a wide rage of distributions. For example, Gaussian
distribution satisfy Condition (C1.2), and Pareto distribution
$Pareto(\alpha)$ (a heavy tail distribution) with
$\alpha$ sufficiently large satisfy Condition (C1.2*).

\noindent {\bf (C1.3)} Let $\theta_{st,ij} =
\Var\{(X_{s,i}-\mu_{s,i})(X_{t,j}-\mu_{t,j})\}$, with $\mu_{s,i}=\E
X_{s,i}$ and $\mu_{t,j}=\E X_{t,j}$. Suppose that there exists $\kappa_1>0$, such that
\[
 \max_{1\leq i\leq q_s, 1\leq j\leq q_t}
 \frac{\sigma_{ss,ii}\sigma_{tt,jj}}{\theta_{st,ij}} \leq \kappa_1.
\]

Condition (C1.3) holds immediately with $\kappa_1=1$ under the null
$\uH_{0,st}$, and thus we only need it for the power analysis.  Under
the alternative $\uH_{1,st}$, it holds for a bunch of distributions. For
instance, it holds when the concatenated vector
$(\bX_{k,s}^\trans,\bX_{k,t}^\trans)^\trans$ follows elliptically
contoured distributions \citep{AndersonIntroduction2003}. In
particular, for multivariate Gaussian distributions, $\kappa_1 \leq 2$.

We first present the asymptotic null distribution of
$T_{st}^{(1)}$. 
\begin{theorem}\label{th:Tnull}
  Suppose that (C1.1) and (C1.2) (or (C1.2*)) hold. Then under
  $\uH_{0,st}$, as $n,q_0\rightarrow \infty$, for all $x \in \RR$,
  the distribution $T_{st}^{(1)}$ converges to the Gumbel
  distribution $F(x)$ defined in (\ref{eq:null_dist}).
\end{theorem}

When (C1.1) is not satisfied, \ie, the correlation matrices
$\Upsilon_{ss}$ and $\Upsilon_{tt}$ are arbitrary, it is difficult to
derive the limiting null distribution of $T_{st}^{(1)}$. However,
Test I can still control the type I error.

\begin{proposition}\label{pr:C2typeI}
 Under (C1.2) (or (C1.2*)) and the null $\uH_{0,st}$, for $0<\alpha<1$, 
\[\P\{T_{st}^{(1)}\geq q_\alpha\}\leq \log\{1/(1-\alpha)\},\]
where $q_\alpha$ is defined in (\ref{eq:qalpha}). 
\end{proposition}

When the desired type I error $\alpha$ is small,
$\log\{1/(1-\alpha)\}\approx \alpha$. Therefore, Test I can still
control type I error close to the desired level. When there comes a
rare circumstance that a larger type I error is desired for the test,
we can define $\alpha' = 1-\exp(-\alpha)$ and reject $\uH_{0,st}$ when
$\widetilde T_{st}^{(1)}\geq q_{\alpha'}$.  Since $\alpha =
\log\{1/(1-\alpha')\}$, Test I is always a asymptotically valid
test, for arbitrary correlation matrices $\Upsilon_{ss}$ and
$\Upsilon_{tt}$. However, the power will be reduced when we threshold
$T_{st}^{(1)}$ at the a higher level $q_{\alpha'}$.

We now turn to the power analysis of Test I. To
test the correlation between region $s$ and region $t$, we define
the following class of correlation matrix:
\[ \Ur_{st}^{(1)} (c) = \left\{ \Upsilon_{st}: n\cdot \max_{i,j} \rho_{st,ij}^2 \geq c\log d_{st}
\right\},\]
It turns out that Test I distinguishes $\Upsilon_{st}$ in
$\Ur_{st}^{(1)}\{4(1+\kappa_1)\}$ from a zero matrix with a probability 
approaching to one asymptotically.

\begin{theorem}\label{th:power}
  Suppose that (C1.2) (or (C1.2*)) and (C1.3) hold. Then as $n$ and
  $q_0$ both go to infinity,
\[\inf_{\Upsilon_{st}\in \Ur_{st}^{(1)}\{4(1+\kappa_1)\}}
\P\{ T_{st}^{(1)} > q_\alpha\} \rightarrow 1.\]
\end{theorem}

To distinguishes the alternative from the null, Test I 
requires only one entry in the
correlation matrix $\Upsilon_{st}$ larger than
$(c\log d_{st}/n)^{1/2}$. The rate is optimal in terms of the
following minimax argument.
Denote by
$\Fr_{st}^{(1)}$ the collection of distributions satisfying (C1.2) or
(C1.2*), and by $\Tr_{st,\alpha}^{(1)}$ the collection of all
$\alpha$-level tests over $\Fr_{st}^{(1)}$, \ie, 
\[\text{For all }
\Phi_{st,\alpha}\in\Tr_{st,\alpha}^{(1)},\quad
\P\{\Phi_{st,\alpha} =
1\}\leq \alpha.
\]
Theorem~\ref{th:lower} shows that, if the maximum absolute correlation is
less than $(c_0\log d_{st}/n)^{1/2}$, for some $c_0$, no test can perfectly
distinguish the alternative from the null. Thus, Theorems~\ref{th:power} and
\ref{th:lower} together indicate that Test I has certain rate optimality
property.

\begin{theorem}\label{th:lower}
  Suppose (C1.2) or (C1.2*) holds. Let $\alpha$ and $\beta$ be any positive
  numbers with $\alpha + \beta <1$. There exists a positive constant $c_0$
  such that for all large $n$ and $q_0$, 
\[\inf_{\Upsilon_{st}\in \Ur_{st}^{(1)}(c_0)} \sup_{T_{st,\alpha}\in
  \Tr_{st,\alpha}^{(1)}} \P(T_{st,\alpha}= 1) \leq 1-\beta.\]
\end{theorem}

In Theorem~\ref{th:power} and \ref{th:lower}, the difference between
the null and the alternative is measured by the maximal absolute value
of the entries in $\Upsilon_{st}$. Another commonly used measure is the
Frobenius norm $\lVert \Upsilon_{st}\rVert_F$. Denote by $r_{st}$ the
count of the nonzero entries in $\Upsilon_{st}$, \ie,
\[r_{st} = \sum_{i=1}^{q_s} \sum_{j=1}^{q_t} I(\rho_{st,ij} \neq 0).\] 
Consider the following class of matrices:
\[\Vr_{st}^{(1)}(c) = \left\{\Upsilon_{st}:\ 
  \lVert \Upsilon_{st}\rVert_F^2 \geq c r_{st}\log d_{st}/n
\right\}.\]
We now show that Test I enjoys the rate optimality
property measured by Frobenius norm too.

\begin{corollary}\label{co:power2}
  Suppose that (C1.2) or (C1.2*) holds. Then for a sufficiently large $c$,
  as $n$ and $q_0$ both go to infinity,
\[\inf_{\Upsilon_{st}\in \Vr_{st}^{(1)}(c)} \P\{ T_{st}^{(1)}> q_\alpha\} \rightarrow 1.\]
\end{corollary}

\begin{theorem}\label{th:lower2}
  Suppose that (C1.2) or (C1.2*) holds. Assume that $r_{st} \leq
  q_0^{\gamma_2}$ for some $0<\gamma_2<1/2$. Let $\alpha,\beta$ be any positive
  number with $\alpha + \beta <1$. There exists a positive contant $c_0$
  such that for all large $n$ and $q_0$, 
  \[\inf_{\bSigma_{st}\in \Vr_{st}^{(1)}(c_0)}\sup_{T_{st,\alpha}\in
    \Tr_{st,\alpha}^{(1)}} \P(\Phi_{st,\alpha}= 1) \leq 1-\beta.\]
\end{theorem}

In Theorem~\ref{th:lower2}, we assume that $r_{st} \leq
q_0^{\gamma_2}$. The assumption is quite reasonable for brain network,
because if the connections of the functional components exist between two
brain regions, they are usually sparse.

\ignore{
Because $\Vr_{st}(c)\subseteq \Ur_{st}(c)$, it is easy to show that
\begin{align*}
  \inf_{\Upsilon_{st}\in \Vr_{st}^{(1)}(c)} \P\{ T_{st}^{(1)}>
  q_\alpha \} & \geq \inf_{\Upsilon_{st}\in
    \Ur_{st}^{(1)}(c)}\P\{\widetilde  T_{st}^{(1)}> q_\alpha)\},\\
  \inf_{\Upsilon_{st}\in \Vr_{st}^{(1)}(c)}\sup_{T_{st,\alpha}\in
    \Tr_{st,\alpha}^{(1)}} \P(T_{st,\alpha}= 1) & \geq
  \inf_{\Upsilon_{st}\in \Ur_{st}^{(1)}(c)}\sup_{T_{st,\alpha}\in
    \Tr_{st,\alpha}^{(1)}} \P(T_{st,\alpha}= 1).
\end{align*}

Therefore, Corollary~\ref{co:power2} follows Theorem~\ref{th:power},
and Theorem~\ref{th:lower} follows Theorem~\ref{th:lower2} directly.
}

\subsection{Asymptotic Properties for Test II}\label{sec:theory_p2}

For Test II, the conditions required for achieving its
asymptotic property are different from what required for Test I.

Recall that $\vareps_{k,s,i}$ and $\vareps_{k,t,j}$ are the error term
of regressing all other components on one component within the region, as
defined in (\ref{eq:ns1}) and (\ref{eq:ns2}), and $\sigma_{\vareps,st,ij} =
\Cov(\vareps_{k,s,i},\vareps_{k,t,j})$. Let $\Upsilon_{\vareps,st} =
(\rho_{\vareps,st,ij})$ be the correlation matrix between $\bvareps_{k,s}
= (\vareps_{k,s,1},\ldots,\vareps_{k,s,q_s})^\trans$ and
$\bvareps_{k,t} = (\vareps_{k,t,1},\ldots,\vareps_{k,t,q_t})^\trans$. Then
\[\rho_{\vareps,st,ij} =
\frac{\sigma_{\vareps,st,ij}}
{(\sigma_{\vareps,ss,ii}\sigma_{\vareps,tt,jj})^{1/2}},\] 
where $\sigma_{\vareps,st,ij} =
\Cov(\vareps_{k,s,i},\vareps_{k,t,j})$, $\sigma_{\vareps,ss,ii} =
\Var(\vareps_{k,s,i})$ and $\sigma_{\vareps,tt,jj} =
\Var(\vareps_{k,t,j})$.

For $\vareps_{k,s,i}$, denote by $r^{(2)}_{v,i}$
the number of other $\vareps_{k,s,j}$ that are non-negligibly
correlated ($>(\log q_0)^{-1-\alpha_0}$) with it,
\[r^{(2)}_{v,i} = \Card\{j:\ \lvert \rho_{\vareps,vv,ij} \rvert
\geq (\log q_0)^{-1-\alpha_0}, \ j\neq i\}.\]
For a positive constant $\rho_0<1$, define the following set that
$\vareps_{k,v,i}$ is highly correlated with 
at least one $\vareps_{k,v,j}$ as 
\[\Dr^{(2)}_v=\{i: \lvert \rho_{\vareps,vv,ij} \rvert >\rho_0
\text{ for some } j\neq i\}.\]

We need the following conditions:

\noindent {\bf (C2.1)} For regions $v=s,t$, there exists a subset $\Mr_v \in
\{1,\ldots,q_v\}$ with $\Card(\Mr_v) = o(q_v)$ and a constant
$\alpha_0>0$ such that all $\gamma>0$, $\max_{1\leq i\leq p, i\in
  \Mr_v} r^{(2)}_{v,i} = o(q_v^\gamma)$. Moreover, assume
there exists a constant $0\leq \rho_0 <1$ such that
$\Card\{\Dr_v\} = o(q_0)$.

Condition (C2.1) parallels with Condition (C1.1). It imposes
conditions on the within region correlation
$\bUps_{\vareps,vv}$. Suppose $\bX_{k,v}$ follow multivariate Gaussian
distribution with $\bOmega_{vv} = (\omega_{vv,ij})$ to be its inverse
covariance matrix. Because
$\rho_{\vareps,vv,ij} =
\omega_{vv,ij}/(\omega_{vv,ii}\omega_{vv,jj})^{1/2}$
\citep{AndersonIntroduction2003}, Condition (C2.1) holds under many
cases when inverse covariance matrix of the components are sparse and
bounded. See \citet{HonorioSparse2009, HuangLearning2010,
  MazumderGraphical2012}. Obviously, the covariance matrix and inverse
covariance matrix are different, and consequently many data only
satisfy one of these two conditions, and then the corresponding
procedure should be applied to the data.

\noindent {\bf (C2.2)} For region $v=s,t$, the variable
$\bX_{k,v}\sim \normal(\bmu_v,\bSigma_{vv})$, with
$\lambda_{\max}(\bSigma_{vv}) \leq c_0$, where $\lambda_{\max}$ is the
maximum eigenvalue operator. Also assume $\log q_0 = o(n^{1/5})$.

In general, the theoretical properties of Test II hold for many
non-Gaussian distributions as well. However, only under the Gaussian
distribution assumption, $\rho_{\vareps,vv,ij}$ has an interpretation
of conditional dependence such that
\[ \rho_{\vareps,vv,ij} = 0 \text{ if and only if } X_{k,v,i}
\independent X_{k,v,j} \mid \{X_{k,v,l},\ l\neq i,j\}. \]
Condition (C2.2)
makes Condition (C2.1) a natrual assumption on the conditional dependency.
Since $\sigma_{vv,ii}\leq \lambda_{\max}(\bSigma_{vv})$ and
$\sigma_{vv,ii}\omega_{vv,ii}\geq 1$, this condition also
implies that $\Var(\vareps_{k,s,i}) = 1/\omega_{vv,ii}\leq c_0$.

\noindent {\bf (C2.3)} Recall the definition of $\tilde{\vareps}_{k,v,l}$ and
$\hat{\vareps}_{k,v,l}$ in (\ref{eq:tdvareps}) and
(\ref{eq:htvareps}). Under the cases (i) $s\neq t$ and (ii)
  $s=t$ and $i=j$, with probability tending to one, 
  \begin{equation}
  \max_{i,j}\left\lvert \frac{1}{n}\sum_{k=1}^n
    \hat{\vareps}_{k,s,i}\hat{\vareps}_{k,t,j} -
    \frac{1}{n}\sum_{k=1}^n
    \tilde{\vareps}_{k,s,i}\tilde{\vareps}_{k,t,j}\right\rvert \leq C(\log
    q_0)^{-1-\alpha_0}.\label{eq:hatvareps}
  \end{equation}

  Note that $\hat{\vareps}_{k,v,i}$ is the centered residual and
  $\tilde{\vareps}_{k,v,i}$ is the centered random error. The term
  $\lvert\frac{1}{n}\sum_{k=1}^n
  \hat{\vareps}_{k,s,i}\hat{\vareps}_{k,t,j} - \frac{1}{n}\sum_{k=1}^n
  \tilde{\vareps}_{k,s,i}\tilde{\vareps}_{k,t,j}\rvert$
  is determined by the difference between $\bbeta_{v,i}$ and its
  estimator $\hat{\bbeta}_{v,i}$. We will specify in
  Section~\ref{sec:estbeta} some estimation methods and corresponding
  sufficient conditions under which Condition (C2.3) will hold.

Theorem~\ref{th:T2null} specifies the null distribution of $T_{st}^{(2)}$.

\begin{theorem}\label{th:T2null}
  Suppose that (C2.1), (C2.2) and (C2.3) hold. Then under $\uH_0$, as
  $n,q_0\rightarrow \infty$, for all $v \in \RR$, $T_{st}^{(2)}$
  weakly converges to the Gumbel distribution $F(x)$ in
  (\ref{eq:null_dist}).
\end{theorem}

The derivation of the limiting null distribution of $T_{st}^{(2)}$
calls for Condition (C2.1); when it is not satisfied, we can still
control type I error based on the following proposition. 

\begin{proposition}\label{pr:T2typeI}
  Under (C2.2) and (C2.3) and the null $\uH_{0,st}$, 
  \[\P\{ T_{st}^{(2)} \geq q_\alpha\} \leq
  \log\left\{1/(1-\alpha)\right\},\]
where $q_\alpha = -log(\pi) - 2\log \log\{1/(1-\alpha)\}$ is
the $(1-\alpha)$-th quantile of $F(x)$ defined in (\ref{eq:null_dist}).
\end{proposition}

The power analysis of Test II parallels to that of Procedure
I. Let $r_{\vareps,st} = \sum_{i=1}^{q_s}\sum_{j=1}^{q_t}
I(\rho_{\vareps,st,ij}
\neq 0)$.
Define the
following two classes of matrices:
\begin{align*}
  \Ur_{st}^{(2)}(c) & = \left\{ \Upsilon_{\vareps,st}:\ \max_{1\leq i
      \leq q_s, 1\leq j\leq q_t} \rho^2_{\vareps,st,ij} \geq c\log
    d_{st}/n
  \right\};\\
  \Vr_{st}^{(2)}(c) & = \left\{\Upsilon_{\vareps,st}:\ \lVert
    \Upsilon_{\vareps,st} \rVert_F^2 \geq c
    r_{\vareps,st} \log d_{st}/n\right\}.
\end{align*}
We have the following theorem.
\begin{theorem}\label{th:p2power}
  Suppose that (C2.2), and (C2.4) hold. Then
 \[\lim_{n,q_0\rightarrow \infty}\inf_{\vR_{st}\in \Ur_{st}^{(2)}(c_1)} \P\left\{ T_{st}^{(2)} \geq
   q_\alpha\right\} = 1,\quad \text{and
 }\lim_{n,p\rightarrow \infty}\inf_{\vR_{st}\in \Vr_{st}^{(2)}(c_2)} \P\left\{ T_{st}^{(2)} \geq
   q_\alpha\right\} = 1,\]
for some $c_2\geq c_1$.
\end{theorem}

Similar as Test I, Test II enjoys certain rate optimality in
its power. Denote by $\Fr_{st}^{(2)}$ the collection of
distributions satisfying (C2.2), and by
$\Tr^{(2)}_{st,\alpha}$ the collection of all $\alpha$-level test over
$\Fr_{st}^{(2)}$.
\begin{theorem}\label{th:p2lower}
  Suppose (C2.2) holds. Let $\alpha,\beta$ be any positive
  number with $\alpha+\beta<1$, There exists a positive constant $c_3$
  such that for all large $n$ and $q_0$,
  \begin{align*}
   \inf_{\vR_{\vareps,st}\in \Ur_{st}^{(2)}(c_3)}
   \sup_{\Phi_{st,\alpha}\in \Tr_{st,\alpha}^{(2)}} \P(\Phi_{st,\alpha} &  =
   1) \leq 1-\beta;\\
   \inf_{\vR_{\vareps,st}\in \Vr_{st}^{(2)}(c_3)}
   \sup_{\Phi_{st,\alpha}\in \Tr_{st,\alpha}^{(2)}} \P(\Phi_{st,\alpha} &  =
   1) \leq 1-\beta.
  \end{align*}
\end{theorem}

\subsection{Asymptotic Properties for Multiple Testing
  Procedure}

The properties of the the multiple testing procedure (\ref{eq:fwer})
are based on the limiting null distribution of each test
statistic. Based on Theorems~\ref{th:Tnull} and \ref{th:T2null}, we have the following results.

\begin{theorem}\label{th:fwer_typeI}
  Consider the multiple testing procedure (\ref{eq:fwer}). If (C1.1)
  and (C1.2) (or (C1.2*)) hold, the procedure
  (\ref{eq:fwer}) with $T_{st}^{(1)}$ controls the family-wise error
  rate at level $\alpha$. If (C2.1) and (C2.2) hold, the procedure
  with $T_{st}^{(2)}$ controls the family-wise error rate at level
  $\alpha$.
\end{theorem}

\section{Estimation of $\widehat{\bbeta}_{v,i}$}\label{sec:estbeta}

Test II depends on the estimators of regression model. Estimating
regression coefficients has been investigated extensively in the past
several decades; methods include the Dantzig selector
\citep{CandesTao07}, the Lasso \citep{TibshiraniRegression1996}, the
SCAD \citep{FanVariable2001}, the adaptive Lasso
\citep{ZouAdaptive2006}, the Scaled-Lasso \citep{SunScaled2012}, the
Square-root Lasso \citep{BelloniSquare2011}, \etc. In this paper, we
focus on the Dantzig selector and Lasso, and discuss when they
will yield good estimators than can be used for our testing
procedures. In particular, we will discuss the necessary conditions
for (C2.3) to hold.

Before we discuss the estimating methods, we introduce the following
notations. For region $v$ and component $i$, let
$\vb_{v,i} =
\frac{1}{n}\sum_{k=1}^n(\bX_{k,v,-i}-\bar{\bX}_{v,-i})^\trans
(X_{k,v,i}-\bar{X}_{v,i})$
be the sample covariance between this components and other components
in the region. Denote by
$\hat{\bSigma}_{vv,-i,-i} = \frac{1}{n}\sum_{k=1}^n
(\bX_{k,v,-i}-\bar{\bX}_{v,-i})(\bX_{k,v,-j}-\bar{\bX}_{v,-j})^\trans$
the sample covariance matrix without component $i$, and let
$\vD_{v,i} = \diag(\hat{\bSigma}_{vv,-i,-i})$. For the following
methods, the tuning parameters are
\[\lambda_{v,i}(\delta) = \delta(\hat{\sigma}_{vv,ii}\log
q_v/n)^{1/2}.\]

\textbf{Dantzig Selector.} For $v=1,\ldots,p$ and $i=1,\ldots,q_v$, the Danztig selector
estimators are obtained by
\begin{equation}
  \label{eq:ds}
  \hat{\bbeta}_{v,i}(\delta) = \arg\min \lvert \balpha \rvert_1,\quad
\text{subject to } \lvert \vD_{v,i}^{-1/2}\hat{\bSigma}_{-i,-i}\balpha
- \vD_{v,i}^{-1/2}\vb_{v,i} \rvert_\infty \leq \lambda_{v,i}(\delta).
\end{equation}

\textbf{Lasso.} For $v=1,\ldots,p$ and $i=1,\ldots,q_v$, the Lasso estimators are obtained by 
\begin{multline}\label{eq:lasso}
  \hat{\bbeta}_{v,i}(\delta) = \vD_{v,i}^{-1/2}\hat{\balpha}_{v,i}(\delta),\\
  \text{where } \hat{\balpha}_{v,i}(\delta) = \arg\min_{\balpha\in\RR^{p-1}}
  \left[ \frac{1}{2n}
    \sum_{k=1}^n\left\{X_{k,v,i}-\bar{X}_{v,i}-
(\bX_{k,v,-i}-\bar{\bX}_{v,-i})\vD_{v,i}^{-1/2}\balpha\right\}^2 +
\lambda_{v,i}(\delta)\lvert \balpha\rvert_1 \right].
\end{multline}

We now demonstrate that under certain conditions, the methods yield
good estimators that satisfy the need to testing.
Define by $a_{v,1}$ and $a_{v,2}$ the error bound
\begin{equation}\label{eq:av}
  a_{v,1} = \max_{1\leq i\leq q_v}\lvert \hat{\bbeta}_{v,i} -
  \bbeta_{v,i} \rvert_1,\quad a_{v,2} = \max_{1\leq i\leq q_v} \lvert \hat{\bbeta}_{v,i} -
  \bbeta_{v,i} \rvert_2  
\end{equation} 

\begin{proposition}\label{pr:ds}
  Suppose that (C2.2) holds. Consider the Dantzig selector estimator
  $\hat{\bbeta}_{v,i}(2)$ in (\ref{eq:ds}). Then if $\max_{1\leq i\leq
    q_v} \lvert \bbeta_{v,i}\rvert_0 = o\left\{ n (\log
    q_0)^{-3-2\alpha_0}[\lambda_{\min}(\bSigma)]^2\right\}$, then
  Condition (C2.3) holds.
\end{proposition}

\begin{proposition}\label{pr:lasso}
  Suppose that (C2.2) holds. Consider the Lasso estiamtor
  $\hat{\bbeta}_{v,i}(2.02)$ in (\ref{eq:lasso}). Then if $\max_{1\leq i\leq
    q_v} \lvert \bbeta_{v,i}\rvert_0 = o\left\{ n (\log
    q_0)^{-3-2\alpha_0}[\lambda_{\min}(\bSigma)]^2\right\}$, Condition (C2.3) holds. 
\end{proposition}

In fact, \propref{pr:ds} holds for any Dantzig selector estimator
$\hat{\bbeta}_{v,i}(\delta)$ with $\delta\geq 2$; and
\propref{pr:lasso} holds for any Lasso estimator
$\hat{\bbeta}_{v,i}(\delta)$ with $\delta>2$.  For computational
simplicity, we chose $\delta=2.02$. In numerical studies, we found
such choice work well in testing.

\section{Simulation Studies}\label{sec:simul}
In this section, we evaluate the performance of the our methods via
two simulation studies: one is focused on the size and power of the
proposed tests for two regions, the other illustrates how to
identity the functional brain network using the proposed tests
under family-wise error rate controls.

\subsection{Size and Power}\label{sec:size_power}
We simulate $\bX_k$, for $k = 1,\ldots, n$, from a normal
distribution with mean zero and covariance $\bSigma_{11,22}$, i.e.
$$\bX_{k} \sim \normal(\bzero_{q_1+q_2}, \bSigma_{11,22}) \quad 
\mbox{ with } \quad\bSigma_{11,22} = \left(\begin{array}{cc}
    \bSigma_{11} & \bSigma_{12}\\
    \bSigma^{\trans}_{12} & \bSigma_{22}
\end{array}\right),$$
where $\bX_k = (\bX^{\trans}_{k,1},\bX^{\trans}_{k,2})^{\trans}$
and $\bX_{k,s}$ is of dimension $q_s$ for $s = 1, 2$.  For comparisons, we also consider a simple test for $\uH_{0,12}$ in \eqref{eq:hypo1} based on
the Person correlation coefficient between the principal
component scores. Specifically, denote by ${\mathbf Z}_s$ the first principal component score of data $(\bX^{\rT}_{1,s},\ldots, \bX^{\rT}_{n,s})^{\rT}$. We compute the sample correlation between ${\mathbf Z}_1$ and ${\mathbf Z}_2$, denoted $\widehat\rho_{12}$. The Fisher's Z transformation is then taken to obtain the testing statistics $T^{(3)}_{12}$ for this simple approach, which is given by
$$T^{(3)}_{12} = \frac{1}{2}\log\left(\frac{1+\widehat\rho_{12}}{1-\widehat\rho_{12}}\right).$$
Using the results by ~\cite{hotelling1953new}, it is straightforward to show that $\sqrt{n-3}T^{(3)}_{12}\rightarrow N(0,1)$ under $\uH_{0,12}$ in \eqref{eq:hypo1}. This implies that we reject $\uH_{0,12}$ if $\sqrt{n-3}|T^{(3)}_{12}|>z_{\alpha/2}$, where $z_{\alpha}$ is the $1-\alpha$ normal quantile. We refer to this testing procedure as test III. 

To define different model specifications on $\bSigma_{11,22}$, we
introduce a few auxiliary matrices. Let $\vA_{d} = (a_{ij})_{d\times
  d}$ where $a_{ii} = 1$ and $a_{ij} \sim 0.5\mbox{Bernoulli}(0.5)$
for $10(k-1)+1\leq i \neq j \leq 10 k$, where $k = 1,\ldots, [d/10]$
and $a_{ij} = 0$ otherwise.  Let $\vB_{d} = (b_{ij})_{d\times d}$
where $b_{ii} = 1$, $b_{i,i+1} = b_{i-1,i} = 0.5$ and $b_{i,j} = 0$
for $|i-j|>3$.

 Let
$\bLam_{d} = (\lambda_{ij})_{d\times d}$ with $\lambda_{ii} \sim
\uU(0.5, 2.5)$ and $\lambda_{ij}=0$ for $i\neq j$.  Now, we define
four different models for $\bSigma_{11}$ and $\bSigma_{22}$.


\begin{itemize}
\item Model 1 (Independent Cases): $\bSigma_{ss} = \bLam_{q_s}$,  for $s = 1,2$.
\item Model 2 (Block Sparse Covariance Matrices): $\bSigma_{ss} =
  \bLam^{1/2}_{q_s} (\vA_{q_s} +\delta_i\vI_{q_s} )/(1+\delta_i)
  \bLam^{1/2}_{q_s}$, for $s = 1,2$, where $\delta_i =
  |\lambda_{\mathrm min}(\vA_{q_s})|+0.05$.
\item Model 3 (Block Sparse Precision Matrices): $\bSigma_{ss} =
  \bLam^{1/2}_{q_s} (\vA^{-1}_{q_s} +\delta^*_i\vI_{q_s}
  )/(1+\delta^*_i) \bLam^{1/2}_{q_s}$, for $s = 1,2$, where
  $\delta^*_i = |\lambda_{\mathrm min}(\vA^{-1}_{q_s})| + 0.05$.
\item Model 4 (Binded Sparse Covariance Matrices): $\bSigma_{ss} =
  \bLam^{1/2}_{q_s} (\vB_{q_s} +\tau_s\vI_{q_s} )/(1+\tau_s)
  \bLam^{1/2}_{q_s}$, for $s = 1,2$, where $\tau_s = |\lambda_{\mathrm
    min}(\vB_{q_s})| + 0.05$.
\item Model 5 (Binded Sparse Precision Matrices): $\bSigma_{ss} =
  \bLam^{1/2}_{q_s} (\vB^{-1}_{q_s} +\tau^*_s\vI_{q_s} )/(1+\tau^*_s)
  \bLam^{1/2}_{q_s}$, for $s = 1,2$, where $\tau^*_s =
  |\lambda_{\mathrm min}(\vB^{-1}_{q_s})| + 0.05$.
\end{itemize}

To simulate the empirical size, we assume
$\bSigma_{12} = \bzero_{q_1\times q_2}$. To evaluate the empirical
power, let $\bSigma_{12} = (\sigma_{ij})_{q_1\times q_2}$ with
$\sigma_{ij} \sim s_{ij}\mbox{Bernoulli}[5/(q_1q_2)]$ with
$s_{ij} \sim \normal(4\sqrt{\log(q_1q_2)/n}, 0.5)$.  The sample size
is taken to be $n = 80$ and $150$, while the dimension $(q_1,q_2)$
varies over $(50,50)$, $(100,150)$, $(200,200)$ and $(250,300)$.  The
nominal significant level for all the tests is set at $\alpha =
0.05$.
The empirical sizes and powers for the five Models, reported in 
Tables \ref{tab:size} and \ref{tab:power}, are estimated from 5,000 replications.


Obviously when the covariance matrix of each region is sparse,
Test I controls the type I error better; and when the precision
matrix is sparse, Test II controls the type I error better. This
implies the essence of condition (C1.1) and (C2.1) when deriving the
limiting null distribution. On the other hand, the simulation also
shows that without these two conditions, there is very little
inflation in the type I error. The power analysis shows the similar
pattern. In general, Test I/II has a larger power when the
covariance/precision matrix is sparse. Both Tests I and II
achieve a much larger power than Test III (Person correlation test on the first PC
scores), although the empirical sizes of Test III are
comparable to the proposed tests.


\begin{table}
\centering
\caption{\label{tab:size}Empirical size of Tests I, II and III for different
  sample sizes 
and models ($\times 10^{-2}$)}
\scriptsize
\begin{tabular}{ccrrrrr}
\toprule
\multirow{2}{*}{Model}&\multirow{2}{*}{Test}&  \multicolumn{5}{c}{$(q_1,q_2)$ }\\ 
\cline{3-7}
& & (30,30) & (50,50) & (100,150) & (200,200) & (300,250) \\
\midrule
& &  &  &$n = 80$  &  &  \\ 
\cline{3-7}
   1& I & 4.50 & 4.46 & 4.54 & 5.14 & 6.16 \\[-1mm]  
   &II & 4.58 & 4.48 & 4.70 & 5.70 & 5.44 \\[-1mm]   
   &III & 6.48 & 6.26 & 3.38 & 5.34 & 7.60 \\[1mm]   
   2& I & 4.20 & 4.60 & 4.52 & 6.04 & 6.06 \\[-1mm]  
      & II& 2.88 & 4.06 & 4.08 & 3.86 & 2.88 \\[-1mm]   
   &III & 6.46 & 4.58 & 8.88 & 7.34 & 6.32 \\[1mm]  
   3& I & 3.44 & 4.02 & 4.50 & 4.98 & 3.20 \\[-1mm]  
      & II& 4.56 & 3.94 & 5.02 & 5.76 & 5.74 \\[-1mm]   
   &III & 8.26 & 3.36 & 7.40 & 6.38 & 3.48 \\[1mm]   
   4& I & 4.80 & 4.82 & 5.12 & 5.22 & 6.02 \\[-1mm]  
      & II& 1.92 & 2.28 & 3.04 & 2.16 & 3.12 \\[-1mm]  
   &III & 4.42 & 3.36 & 6.56 & 4.78 & 3.20 \\[1mm] 
   5& I & 0.88 & 1.02 & 1.06 & 1.90 & 1.90 \\[-1mm]  
      & II& 4.52 & 4.60 & 4.32 & 6.28 & 6.14 \\[-1mm]  
   &III & 4.52 & 4.28 & 5.38 & 4.36 & 6.40 \\  
\midrule
&&  &  &$n = 150$  &  &  \\ 
\cline{3-7}
   1& I &  4.94 & 4.10 & 5.04 & 4.62 & 4.84 \\[-1mm]
    & II & 4.76 & 4.34 & 4.78 & 5.18 & 5.36 \\[-1mm] 
  &III & 8.80 & 4.04 & 6.44 & 5.56 & 5.76 \\[1mm] 
   2& I &  5.08 & 4.62 & 4.48 & 4.88 & 4.74 \\[-1mm] 
    & II & 4.02 & 4.68 & 4.40 & 4.70 & 4.24 \\[-1mm] 
  &III & 5.86 & 7.46 & 3.30 & 4.04 & 5.02 \\[1mm] 
   3& I &  4.94 & 4.68 & 4.50 & 4.86 & 4.60 \\[-1mm] 
    & II & 5.34 & 4.68 & 4.26 & 5.12 & 5.04 \\[-1mm] 
  &III & 2.76 & 8.80 & 4.74 & 5.22 & 3.98 \\[1mm] 
   4& I &  5.02 & 4.78 & 4.96 & 4.92 & 5.10 \\[-1mm] 
    & II & 2.62 & 2.46 & 3.62 & 3.42 & 3.78 \\[-1mm] 
  &III & 2.92 & 5.74 & 6.50 & 5.52 & 4.00 \\[1mm] 
   5& I &  1.96 & 1.92 & 1.96 & 2.18 & 3.10 \\[-1mm]
    & II & 5.62 & 4.46 & 4.04 & 4.92 & 4.94 \\[-1mm]
    &III & 3.38 & 5.92 & 3.90 & 5.42 & 2.34 \\ 
\bottomrule
&   \multicolumn{6}{c}{}\\  
\end{tabular}
\end{table}

\begin{table}
\centering
\caption{\label{tab:power} Empirical power of Tests I, II and III for different
  sample sizes 
and models ($\times 10^{-2}$)}
\scriptsize
\begin{tabular}{ccrrrrr}
\toprule
\multirow{2}{*}{Model}&\multirow{2}{*}{Test}&  \multicolumn{5}{c}{$(q_1,q_2)$ }\\ 
\cline{3-7}
& & (30,30) & (50,50) & (100,150) & (200,200) & (300,250) \\
\midrule
& &  &  &$n = 80$  &  &  \\ 
\cline{3-7}
   1& I & 88.58 & 85.00 & 60.20 & 55.44 & 54.74 \\[-1mm]  
   &II & 88.46 & 85.46 & 60.36 & 55.84 & 54.04 \\[-1mm] 
   &III & 11.32 & 6.26 & 7.06 & 8.66 & 6.18 	\\[1mm] 
   2& I & 88.04 & 80.20 & 59.78 & 55.08 & 55.10 \\[-1mm]   
      & II& 69.72 & 64.10 & 49.70 & 44.72 & 43.94 \\[-1mm]   
   &III & 6.46 & 4.00 & 7.00 & 5.72 & 7.28 \\ [1mm] 
   3& I &69.88 & 65.50 & 50.24 & 44.40 & 44.36 \\[-1mm]   
      & II& 87.46 & 80.40 & 59.30 & 54.94 & 55.90 \\[-1mm]   
   &III & 3.84 & 3.36 & 7.80 & 4.50 & 3.96 \\ [1mm] 
   4& I & 90.24 & 95.42 & 63.40 & 56.08 & 64.32 \\[-1mm]  
      & II& 56.82 & 59.16 & 43.98 & 42.18 & 42.84 \\[-1mm]   
   &III &8.02 & 8.52 & 10.12 & 5.96 & 8.64 \\ [1mm] 
   5& I &80.82 & 75.14 & 44.30 & 35.00 & 34.78 \\[-1mm]  
      & II&  89.94 & 85.36 & 54.30 & 49.90 & 44.96 \\[-1mm]  
   &III &8.12 & 5.30 & 6.52 & 6.68 & 7.60 \\ 
\midrule
&&  &  &$n = 150$  &  &  \\ 
\cline{3-7}
   1& I &  98.82 & 98.08 & 96.66 & 89.24 & 85.22 \\[-1mm] 
    & II & 98.96 & 98.04 & 96.98 & 87.78 & 85.04 \\[-1mm]  
  &III & 13.82 & 4.04 & 8.82 & 7.52 & 9.48 \\[1mm] 
   2& I &  99.14 & 97.86 & 97.02 & 87.62 & 84.46 \\[-1mm]  
    & II &86.98 & 75.92 & 73.30 & 55.58 & 55.18 \\[-1mm]  
  &III & 8.10 & 11.48 & 6.26 & 5.02 & 3.64 \\[1mm] 
   3& I & 90.06 & 87.74 & 76.38 & 54.88 & 55.48 \\[-1mm]  
    & II & 94.58 & 94.70 & 92.48 & 84.80 & 79.94 \\[-1mm]  
  &III &3.80 & 9.26 & 4.26 & 5.84 & 3.04 \\[1mm] 
   4& I &  95.26 & 92.56 & 88.68 & 74.92 & 85.42 \\[-1mm] 
    & II & 85.40 & 67.54 & 64.48 & 58.32 & 59.26 \\[-1mm] 
  &III &9.34 & 10.14 & 9.24 & 6.56 & 6.08 \\[1mm] 
   5& I &  84.74 & 79.74 & 56.00 & 44.96 & 45.40 \\[-1mm] 
    & II & 95.10 & 89.96 & 78.44 & 55.24 & 53.32 \\[-1mm] 
    &III & 7.94 & 9.08 & 5.26 & 3.62 & 2.34 \\
\bottomrule
&   \multicolumn{6}{c}{}\\  
\end{tabular}
\end{table}

\subsection{Network Identifications}\label{sec:netiden}
In this section, we perform the simulation studies to illustrate the
performance of our proposed testing procedure with the family-wise
error rate control on the network identifications.  We simulate a
region-level brain network according to the Erd\"{o}s-R\'{e}nyi model
\citep{ErdosOn1960}. We set the number of regions $p = 90$, and the
probability of any two brain regions being functional connected as
$0.01$.  The simulated brain network is shown in Figure
\ref{fig:simulNet} in the supplementary document.

For every two connected brain regions $s$ and $t$ on the simulated
network, we consider four models that we discussed in Section
\ref{sec:size_power} for the specifications of $\bSigma_{ss}$ and
$\bSigma_{tt}$. Similar to the simulation studies for evaluating the
empirical power, we set $\bSigma_{st} = (\sigma_{ij})_{q_s\times q_t}$
with $\sigma_{ij} \sim s_{ij} \mbox{Bernoulli}(10/d_{st})$ with
$s_{ij} \sim \normal(4 \sqrt{\log(d_{st})/n}, 1)$. We set sample size
$n = 150$ and simulate the fMRI time series based on a normal model,
i.e. $\bX_{k} \sim \normal(\bzero, \bSigma_{q\times q})$, for $k = 1,\ldots,
n$, where $q = \sum_{s=1}^p q_s$ and 
$$\bSigma_{q\times q} = \left(
\begin{array}{cccc}
\bSigma_{11} & \bSigma_{12} & \ldots & \bSigma_{1p}\\
\bSigma_{21} & \bSigma_{12} & \ldots & \bSigma_{2p}\\
\ldots & \ldots & \ldots & \ldots\\
\bSigma_{p1} & \bSigma_{p2} & \ldots & \bSigma_{pp}
\end{array}
\right).$$

Table \ref{tab:acc} reports the accuracy of the network identification
and the performance for multiple testing. Denote $E_{st}$ as the
indicator of the true connectivity between region $s$ and region $t$,
and $\hat{E}_{a,st}$ as the indicator of the estimated connectivity at
the $a$-th iteration, $1\leq s<t \leq p$ and $a=1,\ldots,5000$.  The
\nettpr\ is defined as the percentage of exactly identifying the
correct network, the \fwer\ is the empirical familywise error rate
which is the frequency of having one or mode false discoveries of the
functional connectivity over the brain network, and the \fdr\ is the
empirical false discovery rate which is the proportion of falsely
detecting the functional connectivities among the entire detections.
Mathematically, 
\begin{align*}
  \nettpr & = \frac{1}{5000} \sum_{a=1}^{5000} I(\hat{E}_{a,st} =
  E_{st}, \ \forall\ 1\leq
  s<t\leq p),\\
  \fwer & = \frac{1}{5000}\sum_{a=1}^{5000} I(\hat{E}_{a,st} = 1,
  E_{st} = 0, \
  \exists\ s<t),\\
  \fdr & = \frac{\sum_{a=1}^{5000} \sum_{1\leq s<t\leq p} I
    (\hat{E}_{a,st} = 1, E_{st}=0)}{ \sum_{a=1}^{5000} \sum_{1\leq
      s<t\leq p} I(\hat{E}_{a,st} = 1)}.
\end{align*}
Table \ref{tab:acc} shows the similar pattern as Tables \ref{tab:size}
and \ref{tab:power}. When the covariance matrix is the identity
matrix, Test I performs better than Test II since the
optimization step of Test II introduces extra errors. In
addition, Test I is computationally much faster than Test
II. Therefore we recommend Test I when the covariance matrix is
the identity matrix or sparse, and Test II when the precision
matrix is sparse and its inverse is not sparse.

\begin{table}
  \centering
  \begin{tabular}{crrrcrrr}
    \hline
    &\multicolumn{3}{c}{Test I } && \multicolumn{3}{c}{Test
      II}\\ 
    \hline
    & \nettpr & \fwer & \fdr && \nettpr & \fwer & \fdr\\ 
    \hline
    Model 1 &0.72 &0.02 &0.08 && 0.60 & 0.02 & 0.08 \\ 
    Model 2 &0.64 &0.02 &0.04 &&0.56 & 0.08 & 0.02 \\ 
    Model 3 &0.24 &0.10 &0.06 &&0.68 & 0.04 & 0.12 \\ 
    Model 4 &0.66 &0.04 &0.02  &&0.36 & 0.16 & 0.08 \\
    Model 5  &0.18 &0.12 &0.07  &&0.70 & 0.02 & 0.06 \\
    \hline
\end{tabular}
\caption{Accuracy of the network identification for Tests I and II}
\label{tab:acc}
\end{table}

\section{Application}\label{sec:app}
In this section, we demonstrate our method via an analysis of the
resting-state fMRI data that are collected in the autism brain imaging
data exchange (ABIDE) study \citep{di2013autism}. The major goal of
the ABIDE is to explore the association of brain activity with the
autism spectrum disorder (ASD), which is a widely recognized disease
due to its high prevalence and substantial heterogeneity in
children~\citep{bauman2005neurobiology}. The ABIDE study collected 20
resting-state fMRI data sets from 17 different sites consists of 1,112
individuals with 539 ASDs and 573 age-matched typical controls
(TCs). The resting-state fMRI is a popular non-invasive imaging
technique that measures the blood oxygen level to reflect the resting
brain activity. For each subject, the fMRI signal was recorded for
each voxel in the brain over multiple time points (multiple
scans). The different sites in the ABIDE consortium produced different
number of fMRI scans ranging from 72 to 310. Several regular imaging
preprocessing steps \citep{di2013autism, huettel2004functional}, e.g.,
motion corrections, slice-timing correction, spatial smoothing, have
been applied to the fMRI data, which were registered into the MNI
space (image size: $91\times 109\times 91 (2\mbox{mm}^3))$ consisting
of 228,483 voxels. We concentrate on the network identification over
90 regions in the brain, with regions defined according to the AAL
system.

We take a whitening transformation of original fMRI signals using the
AR(1) model \citep{worsley2002general} to remove the temporal
correlations. The de-trending and de-meaning procedures are also applied for original fMRI signals.  We perform the
principal component analysis (PCA) to summarize the voxel-level fMRI
time series into a relatively small number of principal component
signals within each region.  The number of signals is chosen according
to the criterion of the cumulative variance contribution being larger
than 90\%. The mean number of the principal components over 90 regions
is 18 ranging from 6 to 36.  We apply the proposed methods to identify
the resting state brain network for each subject.  The network for a
group of subjects is defined by including the connections for regions
$i$ and $j$ if they are connected over $85\%$ of subject-level
networks.  The ASD patient and control network include 445 connections
and the 502 connections respectively, where numbers of unique
connections are 31 and 88.  The number of connections shared by both
groups is 441. The control network is denser than the ASD patient
network.  Figure \ref{fig:network} shows the unique connections for
the ASD patient network and the health control network.  In the ASD
patient network, there are two ``hub" brain regions that have at least
4 unique connections to other regions in the brain. They are the
medial part of the superior frontal gyrus (SFGmed-R) and Gyrus rectus
(REC). These regions were demonstrated in the previous references \citep{baron1999social,
  tsatsanis2003reduced, hardan2006abnormal,oblak2011reduced} to be
strongly associated with Autism. Our results suggest that Autism
patients have active region-level functional connectivity
to these three regions, while the controls does not have those
network.  On the other hand, in the health control network, there are
three ``hub" regions that have at least 7 connections. They are the
dorsolateral part of right superior frontal gyrus (SFGdor-R), the left
middle frontal gyrus (MFG-L) and the right middle frontal gyrus
(MFG-R). Our results suggest that the Autism patients break the most
of the connections to these three regions. The brain functions of
these regions are consistent with the Autism clinical symptom. For
example, the superior fontal gyrus is known for being involved in
self-awareness, in coordination with the action of the sensory system
\citep{goldberg2006brain}.  

\begin{figure}[htbp] 
   \centering
   \begin{tabular}{ccc}
   \multicolumn{3}{c}{ASD Patient Brain Network}\\
   \includegraphics[width=1.5in]{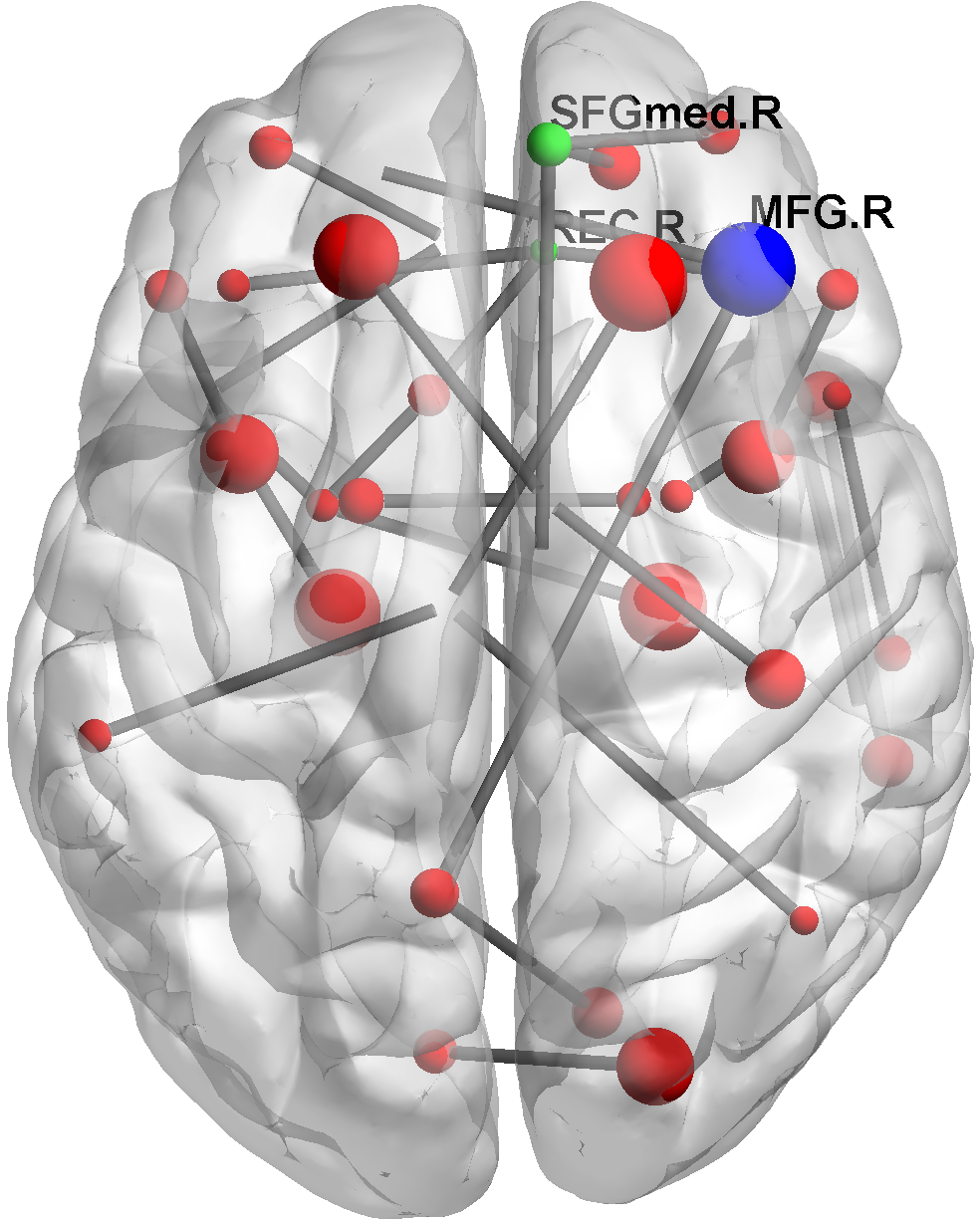} &  
   \includegraphics[width=1.5in]{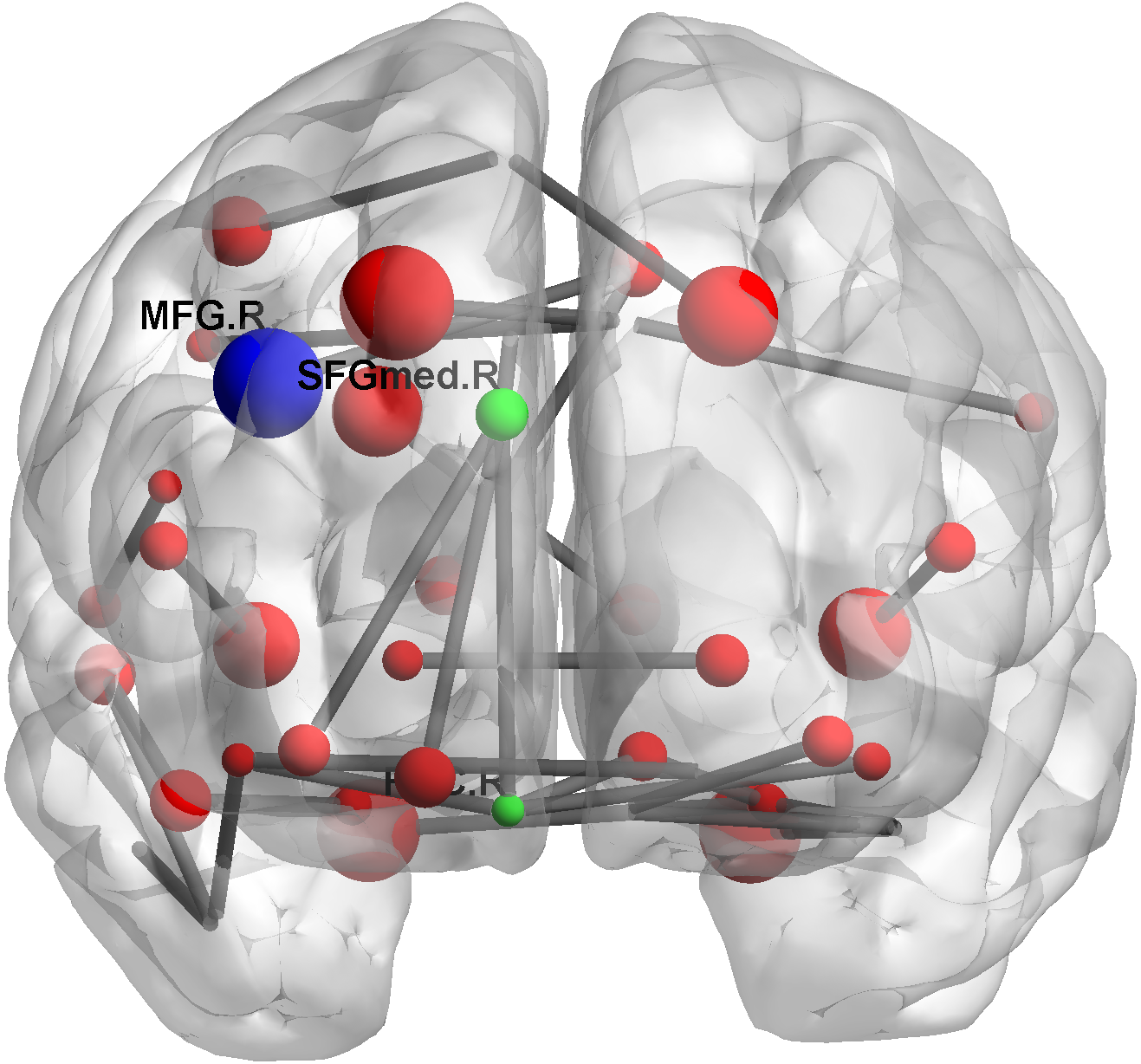}&
   \includegraphics[width=2in]{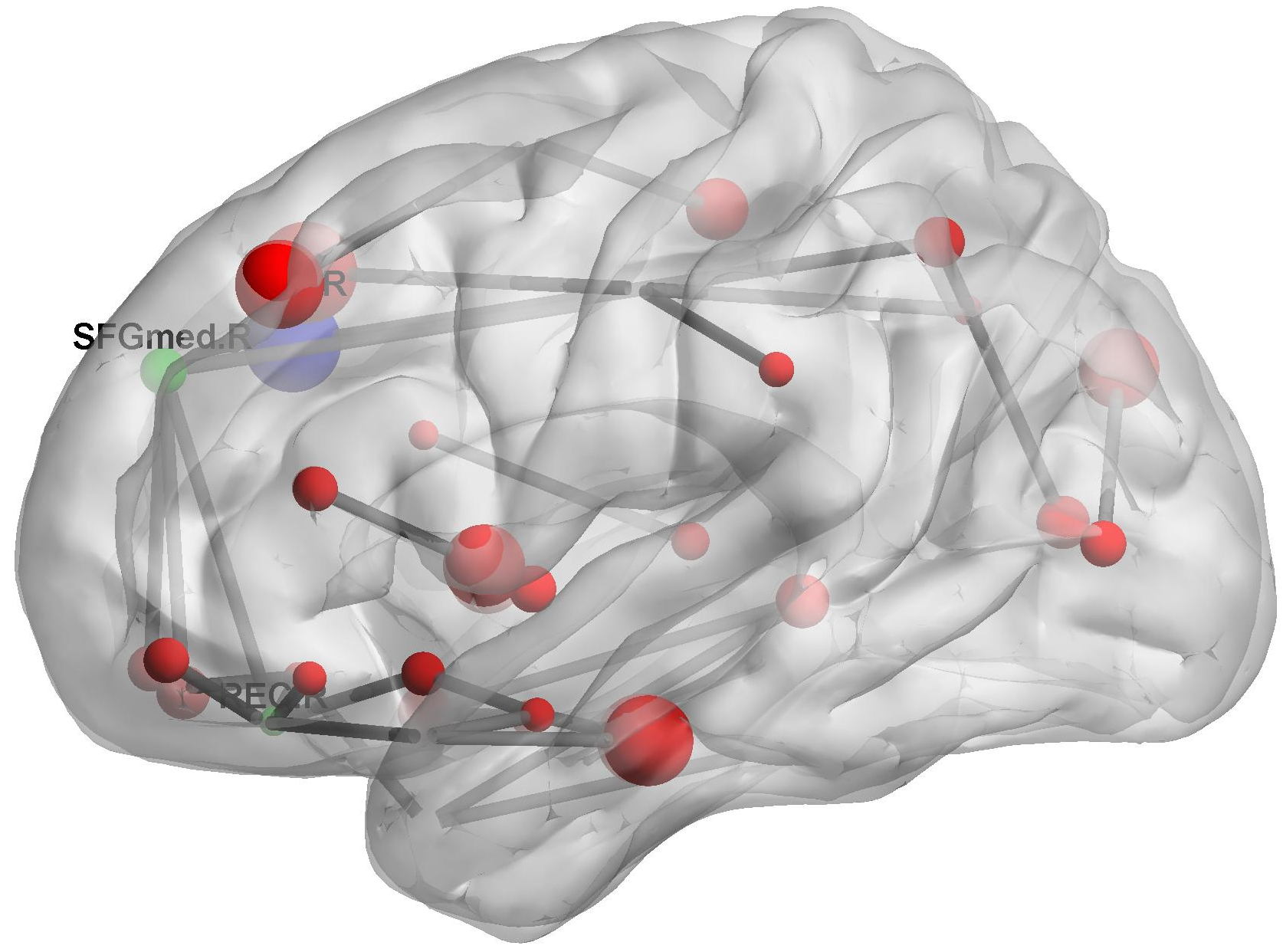}\\
   \multicolumn{3}{c}{Health Control Brain Network}\\
      \includegraphics[width=1.5in]{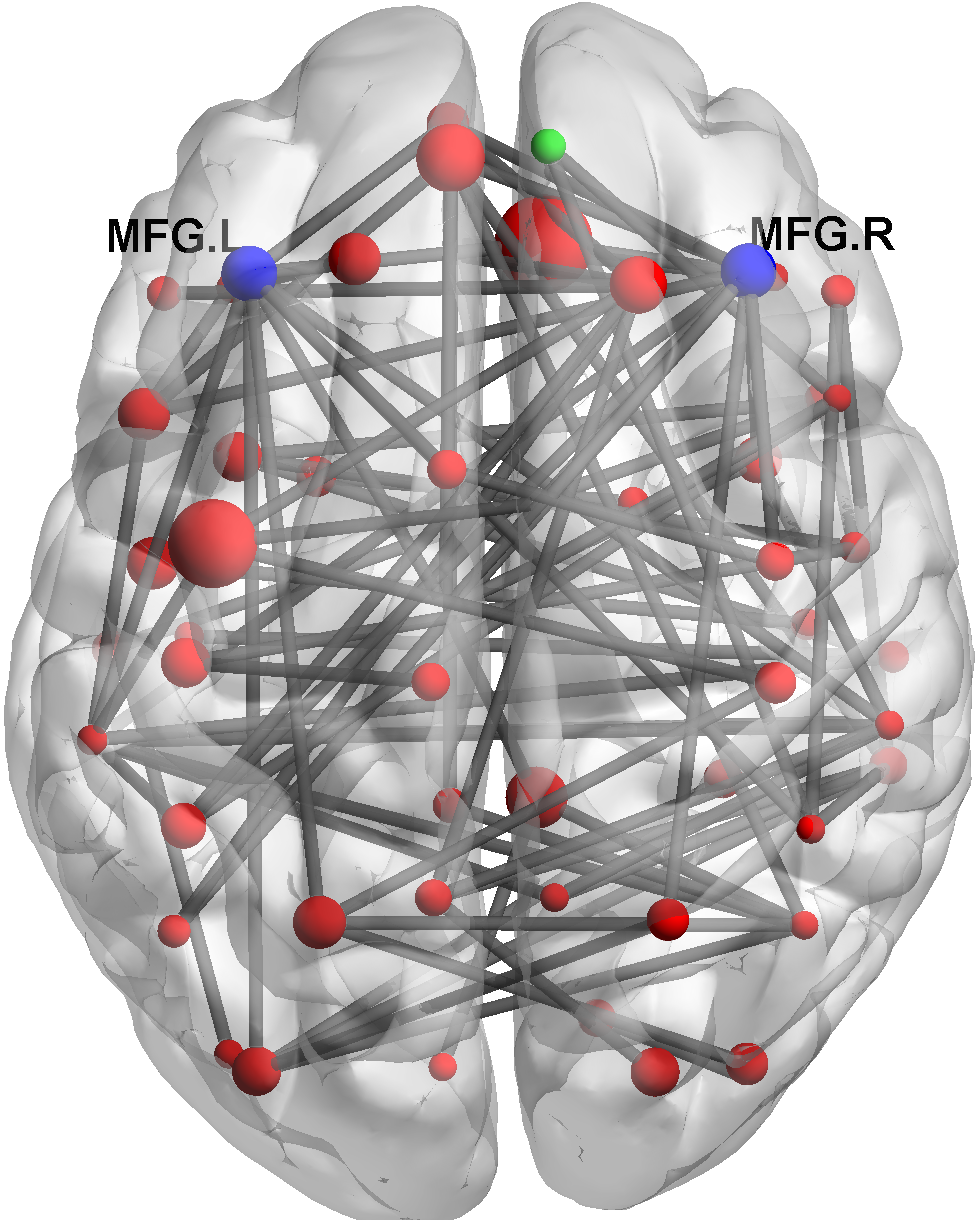} &  
   \includegraphics[width=1.5in]{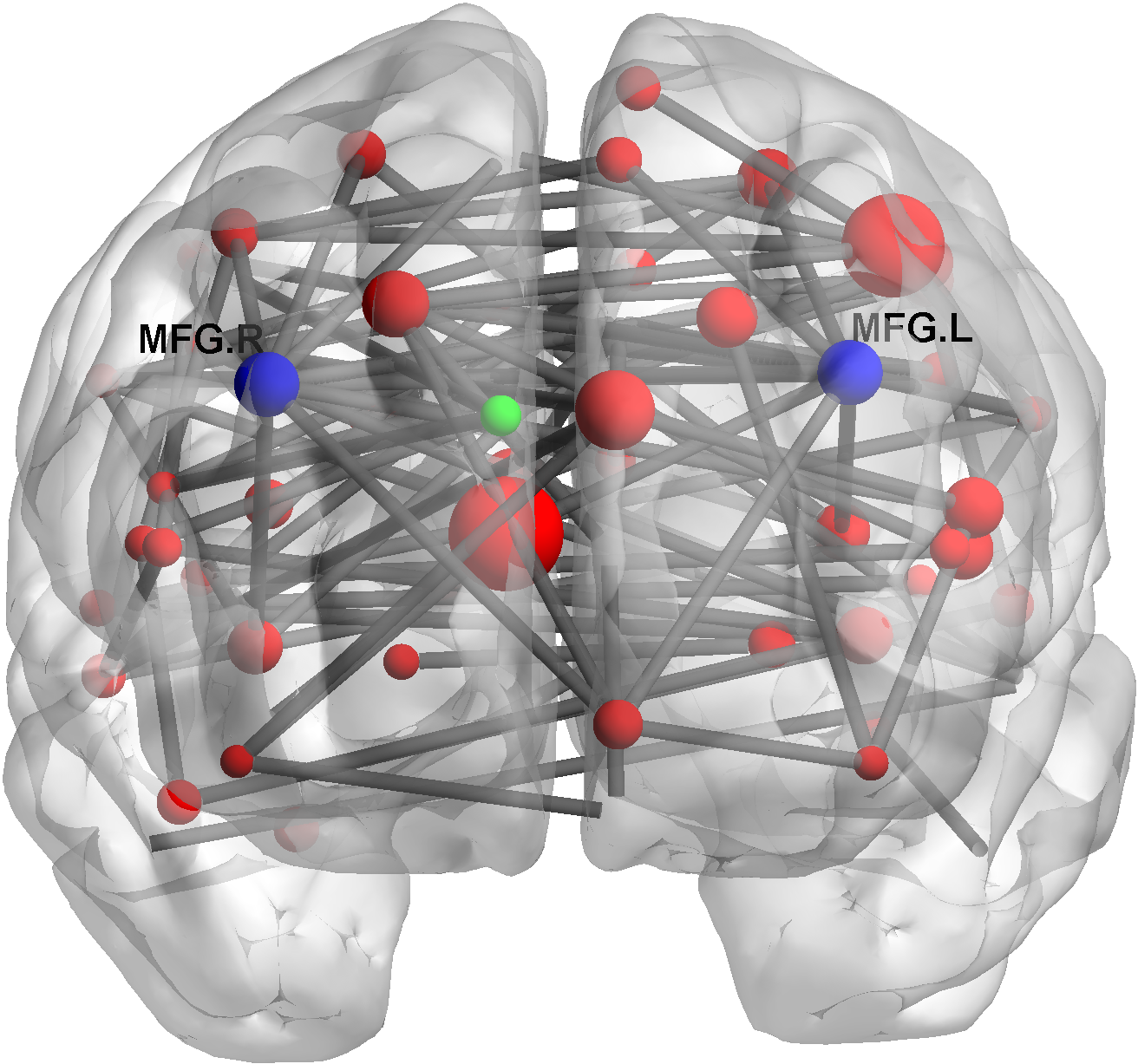}&
   \includegraphics[width=2in]{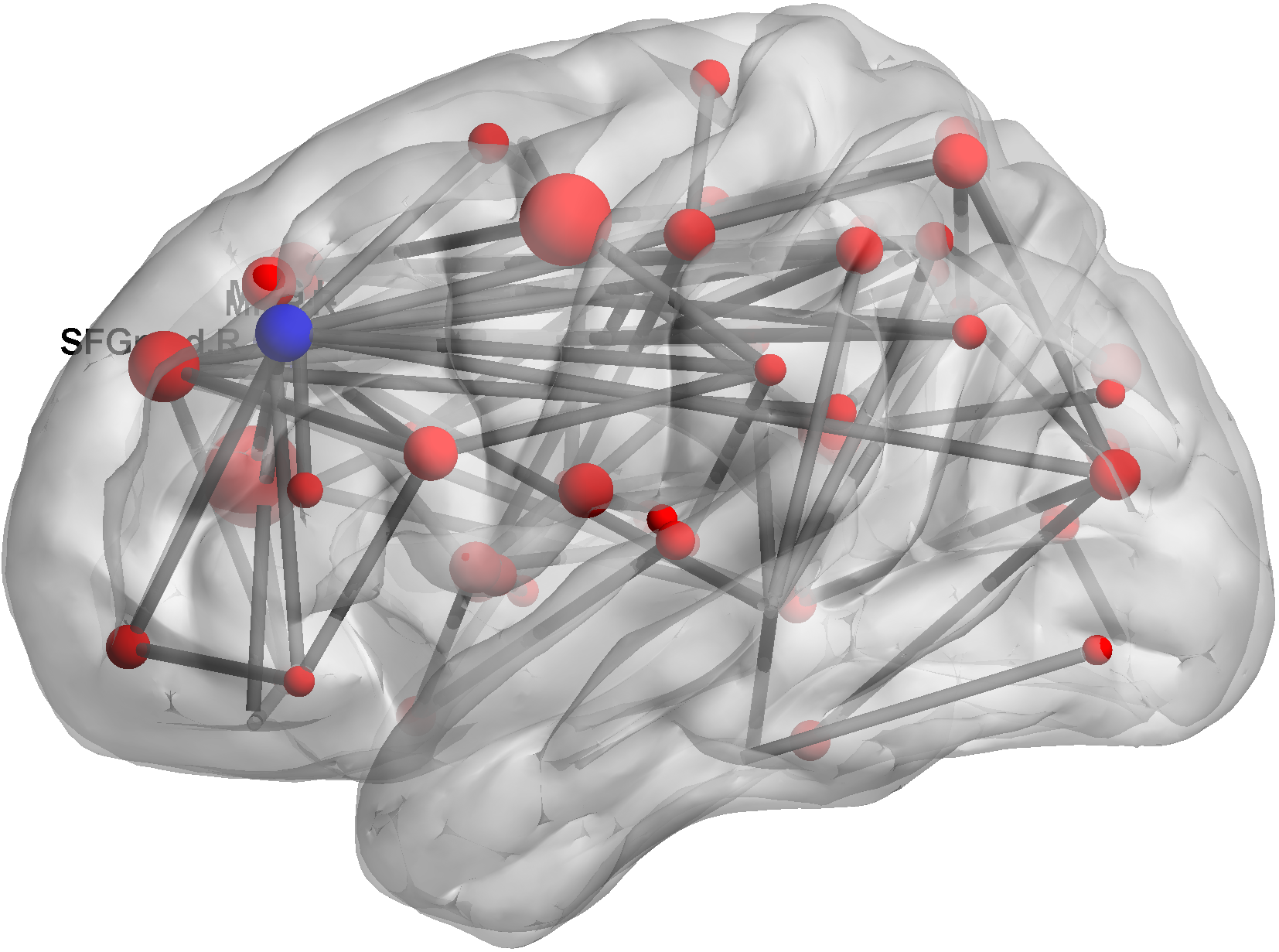}\\

   \end{tabular}
   \caption{Identified region-level resting state brain networks for
     ASD patient group and health control group}
   \label{fig:network}
\end{figure}

%

\section{Discussion}

In additional to this, the novel contributions of
our work include: 1) we propose a new framework to identify the
functional brain network using formal statistical testing procedures,
which make full use of the massive voxel-level brain signals and
incorporate the brain anatomy into the analysis, producing
neurologically more meaningful interpretations.  2) we establish the
statistical theory of the proposed testing procedures, which provides
the solid foundation for making valid inference on the functional
brain network.  3) the proposed method is computationally very
efficient and can be paralleled to achieve fast computing
performance. 4) Although the development of our proposed approach is
motivated by the analysis of brain imaging data, it is a general
method for network construction and can be readily applied to other
problems, such as identification of gene networks and social networks.

\section*{Acknowledgement}

 Jian Kang's research was partially supported by the National Center for Advancing
Translational Sciences of the National Institutes of Health under
Award Number UL1TR000454 and  NIH grant 1R01MH105561. We thank the autism brain imaging data
exchange (ABIDE) study \citep{di2013autism} shares the resting-state
fMRI data.

\section*{Supplementary Material}

The supplementary material includes the proof of 
and all technical lemmas, and the
simulated network (Figure~\ref{fig:simulNet}) on 90 regions using
Er\"{d}os-R\'{e}nyi model discussed in Section~\ref{sec:netiden}.

\appendix
\section*{Proof of Main Theorems}

Without loss of generality, in this section, we assume $\E(X_{k,s,i})
= \E(X_{k,t,j})=0$, and $\Var(X_{k,s,i}) = \Var(X_{k,t,j}) = 1$ unless
otherwise stated. Due to the space limit, we list the proofs of some
theorems (Theorem~\ref{th:power}, Theorem~\ref{th:lower2},
Theorem~\ref{th:T2null}, Proposition~\ref{pr:ds} and
Proposition~\ref{pr:lasso}) here. Theorem~\ref{th:p2power} follows
similar arguments of Theorem~\ref{th:power}, and
Theorem~\ref{th:p2lower} follows that of Theorem~\ref{th:lower2}.  The
proof of Theorem~\ref{th:Tnull} is relatively long and the main
techniques follows the proof of Theorem~1 in \citet{CaiTwo2013}, and
thus is placed in the supplementary material.

In addition, to simplify the notation in the proof, we denote by
$d_{st} = q_sq_t$ the total number of entries in the covariance matrix
$\bUps_{st}$. And also define $c(d_{st},\alpha) = 2\log(d_{st}) -
\log\log(d_{st}) + q_\alpha$, where $q_\alpha$ is the $(1-\alpha)$th
quantile of null distribution $F(x)$.

To prove Theorem~\ref{th:power}, we need Lemma~\ref{lm:theta_dev} and Lemma~\ref{lm:dist}.

\begin{lemma}\label{lm:theta_dev}
  Recall that $\theta_{1,st,ij} = \sigma_{ss,ii}\sigma_{tt,jj}$ and
  $\hat{\theta}_{1,st,ij} = \hat{\sigma}_{ss,ii}\hat{\sigma}_{tt,jj}$.
 Under the conditions of (C1.2) or (C1.2*) and the null $\uH_{0,st}$,
  there exists some constant $C>0$, such that as $n,q_0\rightarrow \infty$,
  \begin{equation}
    \label{eq:thetahtheta0}
    \P\left\{ \max_{i,j} \left\lvert 1 - 
        \frac{\hat{\theta}_{1,st,ij}}{\theta_{1,st,ij}} \right\rvert 
      \geq C\frac{1}{(\log
        q_0)^2} \right\} = O(q_0^{-1} + n^{-\epsilon/4}).
  \end{equation}
\end{lemma}

\begin{lemma}\label{lm:dist}
  Recall that $\theta_{st,ij} = \Var\{(X_{k,s,i}-\mu_{s,i})(X_{k,t,j}-\mu_{t,j})\}$. Under the
  conditions of (C1.2) or (C1.2*), we have for some constant $C>0$
  that
  \begin{equation}
    \label{eq:tildedist}
    \P\left\{ \max_{(i,j)\in \Ar} \frac{(\tilde{\sigma}_{st,ij} -
        \sigma_{st,ij})^2}{\theta_{st,ij}/n} \geq x^2\right\} \leq C\lvert \Ar
    \rvert (1-\Phi(x)) + O(q_0^{-M} + n^{-\epsilon/8})
  \end{equation}
uniformly for $0\leq x\leq (8\log q_0)^{1/2}$ and $\Ar \subseteq \{
(i,j): 1\leq i \leq q_s, 1\leq j\leq q_t \}$. Under $\uH_{0,st}$,
(\ref{eq:tildedist}) also holds when substituting $\theta_{st,ij}$ to $\theta_{1,st,ij}$.
\end{lemma}

\begin{proof}[Proof of Theorem~\ref{th:power}]
Define
\begin{align*}
  T_{st,2} = \max_{i,j}
  \frac{n\hat{\sigma}^2_{st,ij}}{\theta_{1,st,ij}}, & \quad T_{st,3} =
  \max_{i,j} \frac{n\sigma^2_{st,ij}}{\theta_{1,st,ij}},\\
   T_{st,4} =
  \max_{ij}\frac{n(\hat{\sigma}_{st,ij}-\sigma_{st,ij})^2}{\theta_{1,st,ij}},
  & \quad T_{st,5} =   \max_{ij}\frac{n(\hat{\sigma}_{st,ij}-\sigma_{st,ij})^2}{\theta_{st,ij}}.
\end{align*}

By \lemref{lm:theta_dev}, 
\[
  \P(T_{st}^{1} > q_\alpha) \geq \P\{T_{st,2} \geq
  c(d_{st},\alpha)(1+o(1))\}.
\]
Since $T_{st,3} \leq 2 T_{st,4} + 2 T_{st,2}$ and $T_{st,3} \geq
4(1+\kappa_1)\log d_{st}$,
\begin{align*}
   &\ \P\{T_{st,2} \geq c(d_{st},\alpha)(1+o(1))\}\\
  \geq &\  \P\{T_{st,3} -
  2T_{st,4} \geq 2c(d_{st},\alpha)(1+o(1))\} \\
 = &\ \P\{T_{st,4} \leq T_{st,3}/2 - c(d_{st},\alpha)(1+o(1))\}\\
 = &\ \P\{T_{st,4} \leq (2\kappa_1 \log d_{st} + \log\log_{dst} - q_\alpha)(1-o(1))\}.
\end{align*}
By Condition (1.3), $T_{st,5} \geq T_{st,4}/\kappa_1$. It follows that
\begin{multline*}
  \P\{T_{st,4} \leq (2\kappa_1 \log d_{st} + \log\log_{dst} -
  q_\alpha)(1+o(1))\} \\ \geq \P\{ T_{st,5} \leq (2\log d_{st} +
  (1/\kappa_1) \log \log d_{st}-(1/\kappa_1)q_\alpha)(1-o(1)) \}.
\end{multline*}
By \lemref{lm:dist}, 
\[  \P\{ T_{st,5} \leq (2\log d_{st} +
  (1/\kappa_1) \log \log d_{st}-(1/\kappa_1)q_\alpha)(1-o(1)) \}
  \rightarrow 1.\]
\end{proof}

\begin{proof}[Proof of Theorem~\ref{th:lower2}] It suffices to show the
  results for normal distribution which satisfies (C2) and (C2*).
  Denote $\min(q_s,q_t) = q^\ast(s,t)$.  Let $\Mr(s,t)= \{\Sr:\
  \Sr\subseteq\{1,\ldots,q^\ast\},\ \Card(\Sr)=r_{st}\}$ denote
  the set of all the subsets of $\{1,\ldots,q^\ast\}$ with cardinality
  $r_{st}$. Let $\hat{m}$ be a random subset of
  $\{1,\ldots,q^\ast\}$, which is uniformly distributed on
  $\Mr$. Consider such covariance matrix of $(\bX_s, \bX_t)^\trans$:
  \[\bSigma_{\hat{m}}^\ast =
  \begin{pmatrix}
    \vI_{q_s\times q_s}& \bSigma_{st,\hat{m}}^\ast\\
    \bSigma_{st,\hat{m}}^{\ast\trans} & \vI_{q_t\times q_t}
  \end{pmatrix},\quad \text{ and } \bSigma_{st,\hat{m}}^\ast =
  (\sigma_{st,ij})_{q_s\times q_t},
\]
with  
\[\sigma_{st,i_1i_1}=\rho = c(\log d_{st}/n)^{1/2},\ \sigma_{st,i_2i_2} =
\sigma_{st,ij} = 0\quad (i_1\in \Mr(s,t),\ i_2\in \Mr(s,t)^c,\ j\neq
i).\] Here $c$ is a positive constant which will be specified
later. Without loss of generality, suppose $q_s \leq q_t$. Let's
reorder the variables $\bX =
(X_{s,1},X_{t,1},\ldots,X_{s,q_s},X_{t,q_s},\ldots,
X_{t,q_t})^\trans$. Then the covariance matrix of $\bX$ is 
$\bSigma_{\hat{m}} = \diag(A(i),\ldots,A(i),\vI_{q_t-q_s})$, with 
\[
  A(i) =
  \begin{pmatrix}
    1 & \rho\\
   \rho& 1
  \end{pmatrix} \text{ if } i \in \hat{m}; \quad \text{and } A(i)= \vI_2
  \text{ if } i \in \hat{m}^c.
\]
It is easy to see that the precision matrix is $\bOmega_{\hat{m}} =
\diag(B(i),\ldots,B(i),\vI_{q_t-q_s})$, with 
\[
  A(i) = \frac{1}{1-\rho^2}
  \begin{pmatrix}
    1 & -\rho\\
   -\rho& 1
  \end{pmatrix} \text{ if } i \in \hat{m}; \quad \text{and } A(i)= \vI_2
  \text{ if } i \in \hat{m}^c.
\]
We construct a class of $\bSigma$: $\Qr = \{\bSigma_{\hat{m}},
\hat{m}\in \Mr(s,t)\}$. Let $\bSigma_0 = \vI$, and $\bSigma_1$ be
uniformly distributed on $\Qr$. Let $\mu_\rho$ be the distribution of
$\bSigma_1$. It is a measure on $\{\Delta\in \Sr(r_{st},s,t):\ \lVert
\Delta \rVert_F^2 =  r_{st} \rho^2\}$. Let $dP_a(\bX)$ be the likelihood
function given $\bSigma_a$, $a=0,1$. Define
\[L_{\mu_\rho}(\bX) = \E_{\mu_{\rho}}\left\{ \frac{dP_1(\bX)}{dP_0(\bX)} \right\},\] 
where $\E_{\mu_{\rho}}$ is the expectation on $\bSigma_{\hat{m}}$. By
the arguments in Section 7.1 in \citet{BaraudNon2002}, it suffices to
show that $\E_0(L_{\mu_\rho}^2) \leq 1 + o(1)$. 

We have
\[L_{\mu_\rho} = \E_{\hat{m}} \left[\prod_{k=1}^n \frac{1}{\lvert
    \bSigma_{\hat{m}} \rvert^{1/2}} \exp\left\{ -\frac{1}{2}
    \bX_k^\trans(\bOmega_{\hat{m}}-\vI) \bX_k \right\} \right]\]
Let $E_0$ be the expectation on $\bX_k$ with $\normal(0,\vI)$
distribution. Then
\begin{align*}
  \E_0 (L^2_{\mu_{\rho}}) & = \E_0\left[ \frac{1}{{q^\ast \choose r_{st}}}
    \sum_{m\in \Mr}\left\{ \prod_{k=1}^n
      \frac{1}{\lvert\bSigma_m\rvert^{1/2}} \exp\left( -\frac{1}{2}
        \bX_k^\trans(\bOmega_m-\vI) \bX_k \right) \right\}
  \right]^2\\
  & = \frac{1}{{q^\ast \choose r_{st}}^2} \sum_{m,m'\in\Mr} \E_0 \left[
    \prod_{k=1}^n \frac{1}{\lvert \bSigma_m
      \rvert^{1/2}}\frac{1}{\lvert \bSigma_{m'}\rvert^{1/2}} 
    \exp\left\{ -\frac{1}{2}\bX_k^\trans(\bOmega_m +
      \bOmega_{m'}-2\vI) \bX_k   \right\} \right]
\end{align*}
Set $\bOmega_{m} + \bOmega_{m'} - 2\vI = (a_{s_1,s_2,i,j})$,
$s_1,s_2\in\{s,t\}$, $i=1,\ldots,q_{s_1}$, and $j=1,\ldots,q_{s_2}$.
If $i\in m\cap m'$, $a_{ss,ii} = a_{tt,ii} = 2\rho^2/(1-\rho^2)$,
$a_{st,ii} = -2\rho/(1-\rho^2)$. If $i\in m\Delta m'$, $a_{ss,ii} =
a_{tt,ii} = 1/(1-\rho^2)-1$, $a_{st,ii}= -\rho/(1-\rho^2)$. Otherwise,
$a_{s_1,s_2,i,j} = 0$. Now let $t = \lvert m\cap m'\rvert$. By simple
calculations, we have
\begin{align*}
  \E_0(L_{\mu_\rho}^2) & = \frac{1}{{q^\ast \choose r_{st}}^2}
  (1-\rho^2)^{-nr_{st}}\sum_{t=0}^{r_{st}} {q^\ast \choose r_{st}} {r_{st}\choose
    t} {q^\ast-r_{st} \choose r_{st}-t} 1^{tn} (1-\rho^2)^{(2r_{st}-t)n/2}\\
  & = {q^\ast \choose r_{st}}^{-1} \sum_{t=1}^{r_{st}} {r_{st}\choose t}{q^\ast
  -r_{st}\choose r_{st}-t} (1-\rho^2)^{-tn/2}\\
  & \leq q^{\ast r_{st}}\frac{(q^\ast-r_{st})!}{q^\ast !} \sum_{t=0}^{r_{st}}
  {r_{st}\choose
    t}\left(\frac{s}{q^\ast}\right)^t\left(\frac{1}{1-\rho^2}\right)^{tn/2}\\
  & = (1+o(1)) \left( 1+ \frac{r_{st}}{q^\ast
      (1-\rho^2)^{n/2}}\right)^{r_{st}}\\
  & \leq \exp\{r_{st}\log(1+r_{st} q^{\ast c^2-1})\}(1+o(1))\\
  & \leq \exp(r_{st}^2 q^{\ast c^2 -1}) (1+o(1)) 
\end{align*}
For sufficiently small $c^2$, $\E_0(L_{\mu_\rho}^2) = 1+o(1)$, and the
theorem is proved.
\end{proof}

\begin{proof}[Proof of Theorem~\ref{th:T2null}]

Define
\begin{align*}
  T_{st} = n\max_{ij} \rho_{\vareps,st}, & \quad \hat{T}_{st} = \max_{i,j}\frac{n(\hat{\sigma}_{\vareps,st,ij} -
    \sigma_{\vareps,st,ij})^2}{\theta_{\vareps,st,ij}}\\
 \tilde{T}_{st} = \max_{ij}\frac{n(\tilde{\sigma}_{\vareps,st,ij} -
    \sigma_{\vareps,st,ij})^2}{\theta_{\vareps,st,ij}}, & \quad 
 \breve{T}_{st} = \max_{i,j}\frac{n(\breve{\sigma}_{\vareps,st,ij} -
    \sigma_{\vareps,st,ij})^2}{\theta_{\vareps,st,ij}},
\end{align*}
where
\[\hat{\sigma}_{\vareps,st,ij} = \sum_{k=1}^n
      \hat{\vareps}_{k,s,i}\hat{\vareps}_{k,t,j}/n,\quad
      \tilde{\sigma}_{\vareps,st,ij} = \sum_{k=1}^n
      \tilde{\vareps}_{k,s,i}\tilde{\vareps}_{k,t,j}/n, \quad
      \breve{\sigma}_{\vareps,st,,ij} = \sum_{k=1}^n
      \vareps_{k,s,i}\vareps_{k,t,j}/n. \]
By Condition (2.3) and $\max_{i} \lvert
\tilde{\sigma}_{\vareps,ss,ii} - \sigma_{\vareps,ss,ii} \rvert =
O_P\{(\log q_0)^{-1-\alpha_0}\}$,
\begin{multline*}
  \lvert \hat{\theta}_{\vareps,st,ij} - \theta_{\vareps,st,ij}\rvert
  \leq \lvert \hat{\sigma}_{\vareps,ss,ii}\hat{\sigma}_{\vareps,tt,jj}
  - \sigma_{\vareps,ss,ii}\sigma_{\vareps,tt,jj} \rvert \\\leq
  O_P\left\{ \max(\lvert \hat{\sigma}_{\vareps,ss,ii} -
    \sigma_{\vareps,ss,ii} \rvert,\lvert \hat{\sigma}_{\vareps,tt,jj}
    - \sigma_{\vareps,tt,jj} \rvert ) \right\} = O_P\{(\log
  q_0)^{-1-\alpha_0}\}.
\end{multline*}
By (C2.2), 
$\theta_{\vareps,st,ij}\geq 1/c_0^2$. 
Thus with proability tending to one,
\begin{align*}
  \lvert T_{st}-\hat{T}_{st} \rvert & \leq C\hat{T}_{st}(\log q_0)^{-1-\alpha_0}\\
  \lvert \hat{T}_{st} - \tilde{T}_{st} \rvert & \leq  C(\log
  q_0)^{-1-\alpha_0}\\
  \lvert \breve{T}_{st} - \tilde{T}_{st} \rvert & \leq Cn (\max_{1\leq i \leq
    q_s} \bar{\vareps}_{s,i}^4  + \max_{1\leq j \leq
    q_t} \bar{\vareps}_{t,j}^4) + Cn^{1/2} \breve{T}_{st}^{1/2}(\max_{1\leq i \leq
    q_s} \bar{\vareps}_{s,i}^2  + \max_{1\leq j \leq
    q_t} \bar{\vareps}_{t,j}^2).
\end{align*}
The second inequality above is by Condition (C2.3). Note that
\[\max_{1\leq i\leq q_s} \lvert \bar{\vareps}_{s,i} \rvert +
\max_{1\leq t\leq q_t} \lvert \bar{\vareps}_{t,j} \rvert = O_P((\log q_0/n)^{1/2}),\]
Thus, it suffices to show that for any $x\in \RR$,
\[\P\{\breve{T}_{st}\leq 2\log d_{st} - 2\log\log(d_{st}) + x\} \rightarrow
exp\left\{-\frac{1}{\pi^{1/2}}\exp\left(-\frac{x}{2}\right)
\right\}. \]
The rest of the proof is similar to the proof of Theorem~\ref{th:Tnull}.
\end{proof}

\begin{proof}[Proof of \propref{pr:ds}]

We first decompose
$\hat{\sigma}_{\vareps,st,ij}$ as follows:
\[\frac{1}{n}\sum_{k=1}^n \hat{\vareps}_{k,s,i}\hat{\vareps}_{k,t,j} =
\frac{1}{n}\sum_{k=1}^n \tilde{\vareps}_{k,s,i}\tilde{\vareps}_{k,t,j}
- A_{1,s,t,i,j} - A_{2,s,t,i,j} + A_{3,s,t,i,j},\] where
  \begin{align*}
    A_{1,s,t,i,j} & = \frac{1}{n}\sum_{k=1}^n\tilde{\vareps}_{k,s,i}
    (\bX_{k,t,-j}-\bar{\bX}_{t,-j})^\trans
    (\hat{\bbeta}_{t,j}-\bbeta_{t,j}) \\
    A_{2,s,t,i,j} & =  \frac{1}{n}\sum_{k=1}^n\tilde{\vareps}_{k,t,j}
    (\bX_{k,s,-i}-\bar{\bX}_{s,-i})^\trans
    (\hat{\bbeta}_{s,i}-\bbeta_{s,i})\\
    A_{3,s,t,i,j} & =  (\hat{\bbeta}_{s,i}-\bbeta_{s,i})^\trans
   \hat{\bSigma}_{st,-i,-j}(\hat{\bbeta}_{t,j}-\bbeta_{t,j})
  \end{align*}
We bound each term in order. 

Note that for all $s,t\in\{1,\ldots,p\}$,
\begin{align}
\lvert A_{1,s,t,i,j}\rvert \leq & 
\left\vert\frac{1}{n}\sum_{k=1}^n\tilde{\eps}_{k,s,i}
(\bX_{k,t,-j}-\bar{\bX}_{k,t,-j})-\Cov(\tilde{\vareps}_{k,s,i},\bX_{k,t,-j})\right\rvert_\infty\left\lvert
  \hat{\bbeta}_{t,j}-\beta_{t,j} \right\rvert_1 \nonumber\\
& \phantom{a} + \left\lvert\Cov(\tilde{\vareps}_{k,s,i}, \bX_{k,s,-j}^\trans)
\right (\hat{\bbeta}_{t,j}-\bbeta_{t,j})\rvert.\label{eq:a1st}
\end{align}
And also
for any $M>0$,
there exists sufficiently large $C>0$ such that
\[\P\left\{\max_{1\leq i\leq q_s,1\leq j\leq q_t} \left\lvert
    \frac{1}{n} \sum_{k=1}^n \tilde{\vareps}_{k,s,i}(X_{k,t,-j} -
    \bar{X}_{t,-j}) - \Cov(\tilde{\vareps}_{k,s,i},\bX_{k,t,-j})
  \right\rvert_\infty \geq C(\log d_{st}/n)^{1/2} \right\} =
O(q_0^{-M}).\]

Recall the definition of $a_{v,1}$ and $a_{v,2}$ in (\ref{eq:av}).

When $s=t$ and $i=j$, $\Cov(\tilde{\vareps}_{k,s,i},\bX_{k,s,-i})=\bzero$.
Therefore
\[\max_{1\leq i\leq q_s}\left\lvert A_{1,s,s,i,i}\right\rvert = O_{P}\left\{a_{s,1}(\log
  q_s/n)^{1/2}\right\}.\]

When $s\neq t$, under $\uH_{0,st}$,
$\Cov(\tilde{\vareps}_{k,s,i},\bX_{k,t,-j})=\bzero$. Therefore
\[\max_{1\leq i\leq q_s, 1\leq j\leq q_t} \left\lvert 
  A_{1,s,t,i,j}\right\rvert = O_{P}\left\{a_{t,1}(\log
  d_{st}/n)^{1/2}\right\}.\]

When $s\neq t$ and under $\uH_{1,st}$,
\begin{align*}
 \left\lvert\Cov(\tilde{\vareps}_{k,s,i}, \bX_{k,s,-j}^\trans)
\right (\hat{\bbeta}_{t,j}-\bbeta_{t,j})\rvert & \leq 
\left\{\Var(\tilde{\vareps}_{k,s,i})\right\}^{1/2}
\left\{ (\hat{\bbeta}_{t,j}-\bbeta_{t,j})^\trans\bSigma_{tt,-j,-j}(\hat{\bbeta}_{t,j}-\bbeta_{t,j})
\right\}^{1/2}\\
& \leq c_0 a_{t,2}
\end{align*}

Therefore, 
\[\max_{1\leq i\leq q_s, 1\leq j\leq q_t} \left\lvert 
  A_{1,s,t,i,j}\right\rvert = O_{P}\left[ a_{t,1}(\log d_{st}/n)^{1/2}
  + a_{t,2}\right] \] 
We can show bounds for $A_{2,s,t,i,j}$
similarly.

Next, we bound
$A_{3,s,t,i,j}$.  
\begin{align*}
  A_{3,s,t,i,j}  = & (\hat{\bbeta}_{k,s,i}-\bbeta_{k,s,i})^\trans
  (\hat{\bSigma}_{st,-i,-j} -
  \bSigma_{st,-i,-j})(\hat{\bbeta}_{k,t,j}-\bbeta_{k,t,j})\\
   & + (\hat{\bbeta}_{k,s,i}-\bbeta_{k,s,i})^\trans
  \bSigma_{st,-i,-j}(\hat{\bbeta}_{k,t,j}-\bbeta_{k,t,j})
\end{align*}
It is easy to show that for any $M>0$,
there exists sufficiently large $C>0$ such that
\[\P\left\{\max_{1\leq i\leq q_s,1\leq j \leq q_t} \lvert
  \hat{\sigma}_{st,ij} - \sigma_{st,ij}\rvert \geq
  C(\log d_{st}/n)^{1/2}\right\} = O(q_0^{-M}).\]

When $s\neq t$, under $\uH_{0,st}$,
$\bSigma_{st,-i,-j}=\bzero$; and under $\uH_{1,st}$,
$\lVert \bSigma_{st,-i,-j}\rVert_2 \leq c_0$.
By the inequality
\begin{multline}
  \label{eq:a3st}
  \left\lvert (\hat{\bbeta}_{k,s,i}-\bbeta_{k,s,i})^\trans
    (\hat{\bSigma}_{st,-i,-j} -
    \bSigma_{st,-i,-j})(\hat{\bbeta}_{k,t,j}-\bbeta_{k,t,j})\right\rvert
  \\\leq \left\lvert \hat{\bSigma}_{st,-i,-j} - \bSigma_{st,-i,-j}
  \right\rvert_\infty \lvert \hat{\bbeta}_{k,s,i}-\bbeta_{k,s,i}
  \rvert_1 \lvert \hat{\bbeta}_{k,t,j}-\bbeta_{k,t,j}\rvert_1,
\end{multline}
we have under $\uH_{0,st}$,
\[ \max_{1\leq i\leq q_s,1\leq j \leq q_t}\left\lvert A_{3,s,t,i,j}\right\rvert
= O_{P}\left\{a_{s,1}a_{t,1}(\log d_{st}/n)^{1/2}\right\};\]
and under $\uH_{1,st}$, 
\[\max_{1\leq i\leq q_s,1\leq j \leq q_t} \left\lvert A_{3,s,t,i,j}\right\rvert
= O_{P}\left\{a_{s,1}a_{t,1}(\log d_{st}/n)^{1/2} + a_{s,2}a_{t,2}\right\}.\]

When $s=t$, we can show by similar argument that under $\uH_{0,st}$,
\[\max_{1\leq i\leq q_s,1\leq j \leq q_t}\lvert
A_{3,s,s,i,j}\rvert = O_P\left\{ a_{s,1}^2(\log q_s/n)^{1/2} \right\};\]
and under $\uH_{1,st}$,
\[\max_{1\leq i\leq q_s,1\leq j \leq q_t} \lvert
A_{3,s,s,i,j}\rvert = O_P\left\{ a_{s,1}^2(\log q_s/n)^{1/2}+ a_{s,2}^2 \right\}.\]

Therefore, when $s\neq t$, under $\uH_{0,st}$
\begin{equation}
  \frac{1}{n} \sum_{k=1}^n \hat{\vareps}_{k,s,i}\hat{\vareps}_{k,t,j}
    = \frac{1}{n} \sum_{k=1}^n
  \tilde{\vareps}_{k,s,i}\tilde{\vareps}_{k,t,j} + 
  O_P\left\{(a_{s,1}a_{t,1} + a_{s,1} + a_{t,1})\left(\frac{\log
        d_{st}}{n}\right)^{1/2} \right\};\label{eq:vareps1}
\end{equation}
and under $\uH_{1,st}$,
\begin{equation}
\frac{1}{n} \sum_{k=1}^n \hat{\vareps}_{k,s,i}\hat{\vareps}_{k,t,j} =
\frac{1}{n} \sum_{k=1}^n
\tilde{\vareps}_{k,s,i}\tilde{\vareps}_{k,t,j} + O_P\left\{
  (a_{s,1}a_{t,1}+ a_{s,1} + a_{t,1})\left(\frac{\log
        d_{st}}{n}\right)^{1/2} + (a_{s,2}a_{t,2} + a_{s,2} + a_{t,2})
  \right\}.\label{eq:vareps2}
\end{equation} 
When $s=t$ and
$i=j$, under $\uH_{0,st}$,
\begin{equation}
\frac{1}{n} \sum_{k=1}^n \hat{\vareps}_{k,s,i}^2 = \frac{1}{n}
\sum_{k=1}^n \tilde{\vareps}_{k,s,i}^2 + O_P\left\{(a_{s,1}^2 +
  a_{s,1})\left(\frac{\log q_s}{n}\right)^{1/2}\right\};\label{eq:vareps3}
\end{equation}
and under $\uH_{1,st}$,
\begin{equation}
\frac{1}{n} \sum_{k=1}^n \hat{\vareps}_{k,s,i}^2 = \frac{1}{n}
\sum_{k=1}^n \tilde{\vareps}_{k,s,i}^2 + O_P\left\{(a_{s,1}^2 +
  a_{s,1})\left(\frac{\log q_s}{n}\right)^{1/2} + a_{s,2}^2
\right\}.\label{eq:vareps4}
\end{equation}

It then suffices to show that for $v=1,\ldots,p$,
$a_{v,2} = O_P\{ (\log q_0)^{-1-\alpha_0}\}$ and $a_{v,1} =
O_P\{n(\log q_0)^{-2-\alpha_0}\}$.

By the proof of Proposition 4.1 in \citet{LiuGaussian2013}, page 2975,
with probability tending to 1, 
\[\lvert
\vD_{v,i}^{-1/2}\hat{\bSigma}_{vv,-i,-i}\hat{\bbeta}_{v,i}-\vD_{v,i}^{-1/2}\vb_{v,i}\rvert_\infty
\leq \lambda_{v,i}(2).\]
And it follows that
\[\lvert\vD_{v,i}^{-1/2}\hat{\bSigma}_{vv,-i,-i}(\hat{\bbeta}_{v,i}-\bbeta_{v,i})\rvert_\infty
\leq 2 \lambda_{v,i}(2).\]
And also by
\[\max_{1\leq i \leq q_v}\lvert\bbeta_{v,i}\rvert_0  =
o\left\{\lambda_{\min}(\bSigma)(n/\log q_0)^{1/2}\right\}\]
and the inequality 
\[\bdel^\trans\hat{\bSigma}_{vv,-i,-i}\bdel \geq
\lambda_{\min}(\bSigma_{-i,-i})\lvert \bdel\rvert_2^2- O_P\{(\log
q_0/n)^{1/2}\}\lvert \bdel \rvert_1,\]
we can see that the restricted eigenvalue assumption RE$(s,s,1)$ in
\citet{BickelSimultaneous2009}, page 1711, holds
with $\kappa(s,s,1) \geq c\lambda_{\min}(\bSigma)^{1/2}$. And by the
proof of Theorem 7.1 in \citet{BickelSimultaneous2009},
\[a_{v,1} = O_P\left\{ \max_{1\leq i\leq
    q_v}\lvert\bbeta_{v,i}\rvert_0 (\log q_v/n)^{1/2} \right\},\quad 
a_{v,2} = O_P\left[ \left\{\max_{1\leq i\leq
    q_v}\lvert\bbeta_{v,i}\rvert_0 (\log
q_n/n)\right\}^{1/2}\{\lambda_{\min}(\bSigma)\}^{-1}  \right]\]
\end{proof}

\begin{proof}[Proof of Proposition~\ref{pr:lasso}]
  By Proof of Proposition 4.2 in \citet{LiuGaussian2013}, we have with
  probability tending to one, 
  \[\lvert
  \vD_{v,i}^{-1/2}\hat{\bSigma}_{vv,-i,-i}
  \vD_{v,i}^{-1/2}(\hat{\balpha}_{v,i}-\vD_{v,i}^{1/2}\bbeta_{v,i})
  \rvert_\infty \leq 2\lambda_{v,i}(\delta).\] Then by
  (\ref{eq:vareps1}), (\ref{eq:vareps2}), (\ref{eq:vareps3}),
  (\ref{eq:vareps4}), and the proof of Theorem 7.2 in
  \citet{BickelSimultaneous2009}, we get Condition (2.3) holds for
  $\bbeta_{v,i}(\delta)$ with $\delta>2$.
\end{proof}

\newpage

\bibliographystyle{asa}
\bibliography{jichun}



\newpage

\setcounter{section}{0}
\setcounter{equation}{0}
\setcounter{figure}{0}
\setcounter{table}{0}
\renewcommand\theequation{\arabic{equation}}


 
\renewcommand{\thesection}{S.\arabic{section}}
\renewcommand{\thesubsection}{\thesection.\arabic{subsection}}
 
%
\makeatletter 
\def\tagform@#1{\maketag@@@{(S\ignorespaces#1\unskip\@@italiccorr)}}
\makeatother
 
\makeatletter
\makeatletter \renewcommand{\fnum@figure}
{\figurename~S\thefigure}
\makeatother
 
\renewcommand{\bibnumfmt}[1]{[S#1]}
\renewcommand{\citenumfont}[1]{\textit{S#1}}
\renewcommand{\figurename}{Figure}

\setcounter{page}{1}
 


\begin{center}
\textbf{\Large Supplementary Material for ``High Dimensional Tests
  for Functional Brain Networks''}
\end{center}

\section{Proof of Other Theorems}

\begin{lemma}\label{lm:extreme}
  For any fixed integer $D\geq 1$ and real number $x\in \RR$,
\[\sum_{1\leq k_1<\ldots<k_D \leq K} \P\left( \lvert
  \vN_D\rvert_{\min} \geq y(d_{st},x)^{1/2} \pm \eps_n(\log
  q_0)^{-1/2} \right)= \frac{1}{D!} \left\{ \frac{1}{\sqrt{\pi}}
  \exp\left( -\frac{x}{2}\right) \right\}^D (1+o(1)).\]
\end{lemma}

\begin{proof}[Proof of Theorem~\ref{th:Tnull}]
Without loss of generality, we assume that $\mu_{s,i}=\mu_{t,j} = 0$,
$\sigma_{ss,ii} =\sigma_{tt,jj} = 1$, for $i = 1,\ldots,q_s$, and $j =
1,\ldots,q_t$. To simplify notation, let $T =n\cdot \max_{ij}\hat{\rho}_{st,ij}$.

Define
 \[\hat{T} = \max_{i,j} \frac{ (\hat{\sigma}_{st,ij}-\sigma_{st,ij})^2
 }{\theta_{st,ij}/n},\quad \text{and } \tilde{T} = \max_{i,j} \frac{
   (\tilde{\sigma}_{st,ij}-\sigma_{st,ij})^2 }{\theta_{st,ij}/n}\] By
 Lemma~\ref{lm:theta_dev}, with probability at least $1-O(q_0^{-1} +
 n^{-\eps/8})$,
\begin{align*}
  \lvert T - \hat{T}\rvert & \leq C\hat{T} \frac{1}{(\log q_0)^2}\\
  \lvert \hat{T} - \tilde{T} \rvert & \leq \max_{ij} \left\lvert
    \frac{(\hat{\sigma}_{st,ij}-\tilde{\sigma}_{st,ij})(\hat{\sigma}_{st,ij}
      + \tilde{\sigma}_{st,ij} - 2\sigma_{st,ij})}{\theta_{st,ij}/n}
  \right\rvert\\
  & \leq \max_{ij} \left\lvert
    \frac{(\bar{X}_{s,i}\bar{X}_{t,j})\left(2\tilde{\sigma}_{st,ij} -
        2\sigma_{st,ij} - \bar{X}_{s,i}\bar{X}_{t,j}\right)
    }{\theta_{st,ij}/n}
  \right\rvert\\
  & \leq n^{1/2}\tilde{T}^{1/2}\left(\max_i \bar{\bX}_{s,i}^2  + \max_{j}
    \bar{\bX}_{t,j}^2\right) + 2n \left(\max_{i} \bar{\bX}_{s,i}^4 +
    \max_j \bar{\bX}_{t,j}^4\right)
\end{align*}
By similar arguments as (\ref{eq:Xbar1}) and (\ref{eq:Xbar2}),
$\max_i\lvert\bar{X}_{s,i}\rvert + \max_j \lvert \bar{X}_{t,j} \rvert
= O_P\left\{(\log q_0/n)^{1/2}\right\}$. Set $y(d_{st},x) = 2\log
d_{st} - \log\log d_{st} + x$.  By \lemref{lm:dist}, it suffices to
show that for any $x\in \RR$,
\[\P\{\tilde{T} \leq y(d_{st},x)\} \rightarrow
\exp\left\{-\frac{1}{\pi^{1/2}}\exp\left(-\frac{x}{2}\right)\right\}.\]
as $n$ and $d\rightarrow \infty$. 

Let 
\begin{align*}
\Or_{st} & = \{(i,j):\ 1\leq i \leq q_s, 1\leq j\leq q_t\},\\
\Ar_{st} & =\{(i,j): i\not\in \Mr_s, i\not\in \Dr_s^{(1)},  j\not\in
\Mr_t,  j\not\in \Dr_t^{(1)} \}.
\end{align*}
Let 
\[ \tilde{T}_{\Ar_{st}} = \max_{(i,j)\in\Ar_{st}}
\frac{n\tilde{\sigma}_{st,ij}^2}{\theta_{st,ij}},\quad
 \tilde{T}_{\Or_{st}\setminus\Ar_{st}} = \max_{(i,j)\in\Or_{st}\setminus\Ar_{st}}
\frac{n\tilde{\sigma}_{st,ij}^2}{\theta_{st,ij}}.\]
Then 
\[
\lvert \P\{\tilde{T}\geq y(d_{st},x)\} -\P\{\tilde{T}_{\Ar_{st}} \geq y(d_{st},x)\} \rvert
\leq \P\{\tilde{T}_{\Or_{st}\setminus\Ar_{st}} \geq y(d_{st},x)\}.
\] Note that
$\Card(\Or_{st}\setminus \Ar_{st}) = o(d_{st})$. Then by \lemref{lm:dist},
\[\P\left\{\tilde{T}_{\Or_{st}\setminus\Ar_{st}} \geq y(d_{st},x)
\right\} \leq o(d_{st})\cdot C d_{st}^{-1 } + o(1) = o(1). \] It suffies to
show that for any $x\in \RR$,
\[\P\{\tilde{T}_{\Ar_{st}} \leq y(d_{st},x)\} \rightarrow
\exp\left\{-\pi^{-1/2}\exp\left(-x/2\right)\right\}.\]
as $n$ and $q_0\rightarrow \infty$.

We arrange the indices $\{(i,j): (i,j)
\in \Ar_{st}\}$ in any ordering and set them as $\{(i_m,j_m): 1\leq m \leq
d_1\}$, with $d_1 \asymp d_{st}$. Let $\theta_{st,l} = \theta_{st,i_l j_l}$. For
$k=1,\ldots, n$, define
\begin{align*}
  Z_{k,l} & = X_{k,s,i_l} X_{k,t,j_l} - \sigma_{st,i_l j_l},\\
  \hat{Z}_{k,l} & = Z_{k,l}I(\lvert Z_{k,l} \rvert \leq \tau_n) -
  \E\{Z_{k,l}I(\lvert Z_{k,l} \rvert \leq \tau_n)\},\\
  \tilde{Z}_{k,l} & = Z_{k,l}-\hat{Z}_{k,l},\\
  V_l & = \sum_{k=1}^n Z_{k,l}/(n\theta_l)^{1/2},\\
  \hat{V}_l & = \sum_{k=1}^n \hat{Z}_{k,l}/(n\theta_l)^{1/2},\\
  \tilde{V}_l &  = \sum_{k=1}^n \tilde{Z}_{k,l}/(n\theta_l)^{1/2},
\end{align*}
where 
$\tau_n = 8\eta^{-1}\log(d_{st}+n)$ if (C1.2) holds, and $\tau_n
=n^{1/2}/(\log d_{st})^2$ if (C1.2*) holds.
Note that under the null, $\sigma_{st,i_1j_1}=0$. By Markov inequality,
under (C1.2),
\[\P(Z_{k,l} > \tau_n) \leq K_1^2\exp(-\eta/2\tau_n) \leq (d_{st}+n)^{-4},\]
and under (C1.2*),
\[\P(Z_{k,l}>\tau_n) \leq \tau_n^{-4-4\gamma_1-\eps}K_2^2 \leq
C\frac{(\log d_{st})^{8+8\gamma_1+2\eps}}{n^{2+2\gamma_1 + \eps/2}}. \]
The later inequality uses the independence between $X_{k,s,i_l}$ and
$X_{k,t,j_l}$ under $\uH_{0,st}$.

Therefore,
\begin{align}
  \P\left(\max_{1\leq l \leq d_1} \lvert V_l -\hat{V}_l \rvert\geq
    (\log d_{st}+n)^{-M} \right) & = \P \left\{ \max_{1\leq l\leq
      d_1} \lvert \tilde{V}_l\rvert \geq
    (\log d_{st}+n)^{-M} \right\} \label{eq:Vdiff1} \\
  & \leq \P \left( \max_{1\leq l \leq d_1}\max_{1\leq k\leq n}
    \lvert \tilde{Z}_{kl}\rvert > 0 \right) \nonumber\\
  & = n d_{st}\cdot \P(\lvert Z_{kl} \rvert >\tau_n)\nonumber\\
  & \leq O(d_{st}^{-1} + n^{-\eps/4}).\nonumber
\end{align}

By Bernstein's inequality, 
\begin{equation}\label{eq:maxVkhat}
\P\left (\max_{1\leq l\leq d_1} \lvert
\hat{V}_l^2 \rvert\geq (\log d_{st}+n)^2\right) \leq O(d_{st}^{-1}+n^{-\eps})
\end{equation}

It is easy to see that with probability larger than $1-O(d_{st}^{-1} +
n^{-\eps/4})$,
\begin{equation}\label{eq:Vdiff2}
\left\lvert \max_{1\leq l \leq d_1} V_l^2 - \max_{1\leq l\leq d_1}
  \hat{V}_l^2 \right\rvert \leq 2 \max_{1\leq l \leq d_1} \lvert
\hat{V}_l \rvert \max_{1\leq l\leq d_1} \lvert V_l - \hat{V}_l\rvert
+ \max_{1\leq l\leq d_1} \lvert V_l - \hat{V}_l\rvert^2 \leq (\log d_{st}+n)^{-M}.
\end{equation}

It suffices to prove that
for any fixed $x\in \RR$, as $n,d \rightarrow \infty$,
\begin{equation}
  \label{eq:Vdist}
  \P\left\{\max_{1\leq l\leq d_1} \hat{V}_l^2 \leq y(d_{st},x)\right\} \rightarrow
  \exp\left\{ -\pi^{-1/2}\exp(-x/2) \right\}.
\end{equation}

By Bonferroni inequality, for any integer $m$ with $o<m<K/2$,
\begin{multline}
  \label{eq:Bonferroni}
  \sum_{d=1}^{2m} (-1)^{d-1}\sum_{1\leq l_1<\ldots< l_d \leq d_1} 
\P\left(\bigcap_{j=1}^d E_{l_j}\right) \leq \P\left\{\max_{1\leq l\leq d_1}
\hat{V}_l^2 \geq y(d_{st},x)\right\} \\
\leq  \sum_{d=1}^{2m-1} (-1)^{d-1}\sum_{1\leq l_1<\ldots< l_d \leq d_1} 
\P\left(\bigcap_{j=1}^d E_{l_j}\right), 
\end{multline}
where $E_{l_j} = \{\hat{V}_{l_j}^2 \geq y(d_{st},x)\}$.  Let $\vW_{k,d} =
(\hat{Z}_{k,l_1}/\sqrt{\theta_{l_1}}, \ldots,
\hat{Z}_{k,l_d}/\sqrt{\theta_{l_d}})$, for $1\leq k\leq n$. Define
$\lvert \va \rvert_{\min} = \min_{1\leq i \leq d} \lvert a_i \rvert$
for any vector $\va \in \RR^d$. Then,
\[\P\left( \bigcap_{j=1}^d E_{l_j} \right) = \P \left( \big\lvert
n^{-1/2}\sum_{k=1}^n \vW_{k,d} \big\rvert_{\min} \geq y(d_{st},x)^{1/2} \right) \]
By Theorem 1 in \citet{ZaitsevAsymptotic1987}, we have
\begin{multline*}
  \P \left( \big\lvert n^{-1/2}\sum_{k=1}^n \vW_{k,d}
    \big\rvert_{\min} \geq y(d_{st},x)^{1/2} \right) \leq \P\left(
    \lvert \vN_d\rvert_{\min} \geq y(d_{st},x) - \eps_n(\log
    d_{st})^{-1/2} \right) \\
  + c_1d^{5/2} \exp\left( -\frac{n^{1/2}\eps_n}{c_2 d^{5/2} \tau_n
      (\log d_{st})^{1/2}} \right),
\end{multline*}
with $c_1, c_2 >0$ are constants, $\eps_n\rightarrow0$ sufficiently
slow, and $\vN_d$ is a $d-$dimensional normal vector with zero mean
and $\Cov(N_d) = \Cov(\vW_{1,d})$.  Since $d$ is a fixed integer, $\log
q_0 \asymp \log d_{st} = o(n^{1/5})$ and $\eps_n \rightarrow 0 $
sufficiently slow such that
\[ c_1d^{5/2} \exp\left( -\frac{n^{1/2}\eps_n}{c_2 d^{5/2} \tau_n
    (\log d_{st})^{1/2}}\right) = O(q_0^{-M}).\] Thus
\begin{multline*}
  \P\left\{ \max_{1\leq l \leq d_1} \hat{V}_l^2 \geq y(d_{st},x)
  \right\} \\\leq \sum_{d=1}^{2m-1}(-1)^{d-1}\sum_{1\leq l_1<
    \ldots<l_d\leq d_{st}}\P\left\{ \lvert \vN_d\rvert_{\min} \geq
    y(d_{st},x) - \eps_n(\log d_{st})^{-1/2}\right\} + o(1),
\end{multline*}
and similarly 
\begin{multline*}
  \P\left\{ \max_{1\leq l \leq d_1} \hat{V}_l^2 \geq y(d_{st},x)
  \right\} \\\geq \sum_{d=1}^{2m}(-1)^{d-1}\sum_{1\leq l_1<
    \ldots<l_d\leq d_{st}}\P\left\{ \lvert \vN_d\rvert_{\min} \geq
    y(d_{st},x) + \eps_n(\log d_{st})^{-1/2}\right\} - o(1),
\end{multline*}
By \lemref{lm:extreme}, we get
\begin{align*}
  \limsup_{n,q_0\rightarrow \infty} \P\left( \max_{1\leq l\leq d_1}
    \hat{V}_l^2 \right) & \leq \sum_{d=1}^{2m}(-1)^{d-1} \frac{1}{d!}
  \left\{ \frac{1}{\pi^{1/2}}
    \exp\left(-\frac{x}{2}\right)\right\}^d\\
  \liminf_{n,q_0\rightarrow \infty} \P\left( \max_{1\leq l\leq d_1}
    \hat{V}_l^2 \right) & \geq \sum_{d=1}^{2m-1}(-1)^{d-1}
  \frac{1}{d!}
  \left\{ \frac{1}{\pi^{1/2}} \exp\left(-\frac{x}{2}\right)\right\}^d\\
\end{align*}
for any integer $m$. Let $m\rightarrow \infty$, we prove the theorem.
\end{proof}

Without loss of generality, in this section, we assume $\E(X_{k,s,i})
= \E(X_{k,t,j})=0$, and $\Var(X_{k,s,i}) = \Var(X_{k,t,j}) = 1$ unless
otherwise stated.

\begin{proof}[Proof of Proposition~\ref{pr:C2typeI}]
Define $T_{st,ij}^{(1)} =  n \hat{\rho}_{st,ij}$.
By the proof of Theorem~\ref{th:Tnull}, under (C2) (or (C2*)), we have 
\begin{align*}
  & \P_{\uH_0}\{T_{st,ij}^{(1)} > q_\alpha + 2\log d_{st} -
  \log\log d_{st}\} \\
  = & (1+o(1)) \P(\lvert N_1 \rvert \geq q_\alpha + 2\log d_{st} -
  \log\log d_{st}) \\
  = & (1+o(1)) \frac{1}{d_{st}}\log\left(\frac{1}{1-\alpha}\right).
\end{align*}
Note that $T_{st}^{(1)}=\max_{i,j}
T_{st,ij}^{(1)} - 2\log(d_{st}) + \log\log(d_{st})$. Then
\[\P_{\uH_0}\{T_{st}^{(1)}>q(\alpha)\} \leq d_{st}\cdot
\P_{\uH_0} \{T_{st,ij}^{(1)}\geq c(d_{st},\alpha)\} \leq \log\left(
  \frac{1}{1-\alpha}\right).\]
\end{proof}

\begin{proof}[Proof of \lemref{lm:theta_dev}]
  Under $\uH_{0,st}$, $\theta_{st,ij} = \sigma_{ss,ii}\sigma_{tt,jj}$
  and $\hat{\theta}_{st,ij} =
  \hat{\sigma}_{ss,ii}\hat{\sigma}_{tt,jj}$. Thus
\[\frac{\lvert\hat{\theta}_{st,ij} -
  \theta_{st,ij}\rvert}{\sigma_{ss,ii}\sigma_{tt,jj}} \leq
\left\lvert\frac{\hat{\sigma}_{ss,ii}}{\sigma_{ss,ii}}-1\right\rvert
\cdot \left\lvert
  \frac{\hat{\sigma}_{tt,jj}}{\sigma_{tt,jj}}\right\rvert +
\left\lvert \frac{\hat{\sigma}_{tt,jj}}{\sigma_{tt,jj}} -
  1\right\rvert\]

It suffices to show that  
\begin{equation}
  \P \left\{ \max_i \left\lvert \frac{\hat{\sigma}_{ss,ii}}{\sigma_{ss,ii}}-1
    \right\rvert \geq \frac{C}{3}\frac{1}{(\log q_0)^2} \right\} =
  O(q_0^{-1}+n^{-\eps/8}),\label{eq:hsig}
\end{equation}
 and the same holds for $\hat{\sigma}_{tt,jj}$.

Without loss of generality, we assume that $\mu_{s,i}=\mu_{t,j} = 0$,
$\sigma_{ss,ii} =\sigma_{tt,jj} = 1$, for $i = 1,\ldots,q_s$, and $j =
1,\ldots,q_t$.
We have
\[ \frac{\hat{\sigma}_{ss,ii}}{\sigma_{ss,ii}}-1 =
\frac{1}{n}\sum_{k=1}^n \left\{ \bX_{k,s,i}^2 - \E(\bX_{k,s,i}^2)
\right\} - (\bar{\bX}_{s,i})^2 \]

 We first prove the results under
(C1.2). Define $Y_{k,s,i} = X_{k,s,i}^2 - \E(X_{k,s,i}^2)$. Then
\begin{align*}
  & \P\left\{ \max_i \left\lvert
      \frac{\hat{\sigma}_{ss,ii}}{\sigma_{ss,ii}}-1 \right\rvert \geq
    \frac{C}{3}\frac{1}{(\log q_0)^2} \right\}\\
  \leq & \P\left\{\max_i \left\lvert \frac{1}{n} \sum_{k=1}^n
      Y_{k,s,i} \right\rvert \geq \frac{C}{6}\frac{1}{(\log q_0)^2}
  \right\} + \P\left\{\max_i (\bar{\bX}_{s,i})^2 \geq \frac{C}{6}
    \frac{1}{(\log q_0)^2}\right\}\\
  \leq & q_0 \cdot \P\left\{\left\lvert \frac{1}{n} \sum_{k=1}^n
      Y_{k,s,i} \right\rvert \geq \frac{C}{6}\frac{1}{(\log
      q_0)^2} \right\} + q_0 \cdot \P\left\{ \bar{\bX}_{s,i} \geq
    \left(\frac{C}{6} \frac{\vareps_n}{(\log q_0)^2}\right)^{1/2}\right\}
\end{align*}
Let $t_1 = \eta(\log q_0)^{1/2}/(2n^{1/2})$. Then
we have
\begin{align}
  & \P\left\{\left\lvert \frac{1}{n} \sum_{k=1}^n Y_{k,s,i}
    \right\rvert \geq
    \frac{C}{6}\frac{1}{(\log q_0)^2} \right\}\label{eq:Xsq1}\\
  \leq & \exp\{-Ct_1n/(6(\log q_0)^2)\}\cdot \E\left[ \exp\left\{
      \sum_{k=1}^n t_1\lvert Y_{k,s,i}\rvert
    \right\} \right]\nonumber\\
  \leq & \exp\{-Ct_1n/(6(\log q_0)^2)\} \cdot \prod_{k=1}^n
  \E\left\{ \exp(t_1 \lvert Y_{k,s,i} \rvert) \right\}\nonumber\\
  \leq & \exp\{-Ct_1n\vareps_n/(6\log q_0)\} \cdot \prod_{k=1}^n\left[
    1 + \E\left\{ t_1^2Y_{k,s,i}^2\exp(t_1\lvert Y_{k,s,i}\rvert)\right\}\right]\\
  \leq & \exp\left[-Ct_1n/(6(\log q_0)^2) + \sum_{k=1}^n \E\left\{
      t_1^2Y_{k,s,i}^2\exp(t_1 \lvert Y_{k,s,i} \rvert)\right\}\right]\nonumber\\
  \leq & \exp(-C \eta\log q_0/12 + c_{\eta} \log q_0)\nonumber\\
  \leq & Cq_0^{-M}, \nonumber
\end{align}
where $c_{\eta}$ is a positive number only depends on
$\eta$. Similarly,
\begin{align}
  & \P\left\{ \bar{\bX}_{s,i} \geq
    \left( \frac{C}{6} \frac{1}{(\log q_0)^2} \right)^{1/2} \right\} \label{eq:Xbar1}\\
  \leq & \exp\left\{ - \frac{\eta}{2}\left(\frac{Cn}{6(\log
        q_0)^2}\right)^{1/2} + c_\eta \log q_0
  \right\} \nonumber\\
  \leq & Cq_0^{-M}\nonumber
\end{align}

It remains to prove the lemma under (C1.2*). Define
\[\hat{Y}_{k,s,i} = Y_{k,s,i}I\left\{\lvert Y_{k,s,i}\rvert \leq
  n/(\log q_0)^5\right\} - \E\left[ Y_{k,s,i}I\left\{\lvert
    Y_{k,s,i}\rvert \leq n/(\log q_0)^5\right\} \right].\]
Then,
\begin{align}
  & \P\left\{\max_i \left\lvert \sum_{k=1}^n Y_{k,s,i} \right\rvert
    \geq
    \frac{C}{6} \frac{1}{(\log q_0)^2}\right\} \label{eq:Xsq2}\\
  \leq & \P\left\{\max_i \left\lvert \sum_{k=1}^n \hat{Y}_{k,s,i}
    \right\rvert \geq \frac{C}{6} \frac{1}{(\log q_0)^2}\right\} +
  \P\left\{\max_{i,k} \lvert Y_{k,s,i} \rvert \geq \frac{n}{(\log q_0)^5}\right\}\nonumber\\
  \leq & Cq_0 \exp\left\{-C (\log q_0)^2\right\} + Cn^{-\eps/4}.\nonumber
\end{align}
The last inequality is by Bernstein's inequality and condition
(C1.2*). Define 
\[\hat{X}_{k,s,i} = X_{k,s,i}I(\lvert X_{k,s,i}-\bar{X}_{s,i}\rvert
\leq n/(\log q_0)^5) - \E\left\{ X_{k,s,i}I(\lvert
  X_{k,s,i}-\bar{X}_{s,i}\rvert \leq n/(\log q_0)^5)\right\}.\] Then,
following the similar argument, we have
\begin{align}
  & \P\left\{ \max_i \bar{\bX}_{s,i} \geq \left(\frac{C}{6}
      \frac{1}{(\log q_0)^2}\right)^{1/2}\right\}\label{eq:Xbar2}\\
  \leq & \P\left\{ \max_i \left\lvert \sum_{k=1}^n \hat{X}_{k,s,i}
    \right\rvert \geq n \left(\frac{C}{6(\log
        q_0)^2}\right)^{1/2}\right\} + \P\left\{\max_{i,k} \lvert
    X_{k,s,i} \rvert \geq n/(\log q_0)^5\right\}\nonumber\\
\leq & C q_0 \exp\left\{ -C(\log q_0)^4\right\} +
Cn^{-2-2\gamma_1-\eps/2}.
\nonumber
\end{align}
\end{proof}

\begin{proof}[Proof of \lemref{lm:dist}]
  Set $Y_{k,st,ij} = X_{k,s,i}X_{k,t,j} - \sigma_{st,ij}$. Define
  $\tilde{\theta}_{st,ij} = \frac{1}{n}\sum_{k=1}^n Y_{k,st,ij}^2$ as an oracle estimator of
  $\theta_{st,ij} = \Var(X_{k,s,i}X_{k,t,j})$. By the proof of Lemma
  4 in \citet{CaiTwo2013}, it follows that
  \begin{equation}
  \P \left(\max_{ij} \left\lvert \frac{1}{n} \sum_{k=1}^n
     Y_{k,st,ij}^2-\theta_{st,ij} \right\rvert \geq C\vareps_n/\log
   q_0 \right) = O(q_0^{-M} + n^{-\eps/8}),\label{eq:tildetheta}
\end{equation}
where $\vareps_n = \max\{(\log q_0)^{1/6}/n^{1/2}, (\log
q_0)^{-1}\}$. We can write
\[\frac{(\tilde{\sigma}_{st,ij} -
  \sigma_{st,ij})^2}{\theta_{st,ij}/n} = \frac{(\sum_{k=1}^n
  Y_{k,st,ij})^2}{\sum_{k=1}^n Y_{k,st,ij}^2}\cdot \frac{\sum_{k=1}^n
  Y_{k,st,ij}^2}{\theta_{st,ij}/n}.\]
By Theorem 1 in \citet{JingSelf2003}, we have
\[\max_{i,j} \P\left\{ \frac{(\sum_{k=1}^n
    Y_{k,st,ij})^2}{\sum_{k=1}^n Y_{k,st,ij}^2}\geq x^2  \right\}
\leq C(1-\Phi(x)).\]
Together with (\ref{eq:tildetheta}), we have the conclusion. 
Note that under the null, $\theta_{1,st,ij} = \theta_{st,ij}$. So
(\ref{eq:tildetheta}) also holds for $\theta_{st,ij}$ under $\uH_{0,st}$.
\end{proof}

\begin{proof}[Proof of \lemref{lm:extreme}]
  When $d=1$, it is easy to get
  \[\P\left( \lvert \vN_1 \rvert_{\min} \geq y(d_{st},x)^{1/2} \pm \eps_n(\log
    d_{st})^{-1/2} \right) = \frac{1}{d_{st}\pi^{1/2}}\exp(-x/2)
  (1+o(1)). \] We now prove the lemma for $d\geq 2$. Note that for any
  $1\leq i,j \leq q_s $ and $1\leq k,l\leq q_t$, under $\uH_{0,st}$,
  we have
\[\Cov(X_{s,i}X_{t,k},X_{s,j}X_{t,l}) =
\sigma_{ss,ij}{\sigma_{tt,kl}}.\] To simplify notation, denote
$X_{s,i}$ by $X_{i_{m_1}}$， $X_{s,j}$ by $X_{i_{m_2}}$, $X_{t,k}$ by
$X_{j_{m_1}}$, and $X_{t,l}$ by $X_{j_{m_2}}$.  Define graph
$G_{i_{m_1}j_{m_1}i_{m_2}j_{m_2}} = (V_{i_{m_1}j_{m_1}i_{m_2}j_{m_2}},
E_{i_{m_1}j_{m_1}i_{m_2}j_{m_2}})$, where
$V_{i_{m_1}j_{m_1}i_{m_2}j_{m_2}} =
\{i_{m_1},j_{m_1},i_{m_2},j_{m_2}\}$ is the set of vertices and
$E_{i_{m_1}j_{m_1}i_{m_2}j_{m_2}}$ is the set of edges. There is an
edge between $ a \neq b \in \{i_{m_1},j_{m_1},i_{m_2},j_{m_2}\}$ if
and only if $\lvert \rho_{ss,ij}\rvert =\lvert \rho_{i_{m_1}i_{m_2}}
\rvert \geq (\log q_0)^{-1-\alpha_0}$ or $\lvert \rho_{tt,kl} \rvert =
\lvert\rho_{j_{m_1}j_{m_2}}\rvert \geq (\log q_0)^{-1-\alpha_0}$, for
all
$a,b\in\{i_{m_1},j_{m_1},i_{m_2},j_{m_2}\}$. $G_{i_{m_1}j_{m_1}i_{m_2}j_{m_2}}$
is a $v$ vertices graph ($v$-G) if the number of different vertices in
$V_{i_{m_1}j_{m_1}i_{m_2}j_{m_2}}$ is $v$. It is a $e$ edges graph
($e$-E) if $\Card(E_{i_{m_1}j_{m_1}i_{m_2}j_{m_2}})=e$. A vertex in
$G_{i_{m_1}j_{m_1}i_{m_2}j_{m_2}}$ is said to be isolated if there is
no edge connected to it. Note that for any $1\leq m_1 \neq m_2 \leq
d$, $G_{i_{m_1} j_{m_1} i_{m_2} j_{m_2}}$ could only be 3G/4G, and
0E/1E/2E. We say a graph $G = G_{i_{m_1}j_{m_1}i_{m_2}j_{m_2}}$
satisfies the weak correlation condition (\ref{eq:wk}) if
\begin{equation}
G \text{ is a 3G0E, 4G0E or 4G1E}\label{eq:wk}.
\end{equation}
For any $G_{i_{m_1} j_{m_1} i_{m_2} j_{m_2}}$ satisfying Condition
(\ref{eq:wk})
\[ \lvert \Cov(X_{i_{m_1}}X_{j_{m1}},X_{i_{m_2}}X_{j_{m_2}})\rvert
= O\{(\log d)^{-1-\alpha_0}\}.\]
We now define the following set
\begin{align*}
  \Ir  = &\left\{ 1\leq k_1 < \ldots < k_d \leq d_{st}\right\},\\
  \Ir_0 = &\left\{ 1\leq k_1 < \ldots < k_d \leq d_{st}: \text{ for some }
    m_1, m_2 \in \{k_1,\ldots,k_d\} \text{ with } m_1\ldots m_2
  \right.\\
  & \left. G_{i_{m_1} j_{m_1} i_{m_2} j_{m_2}} \text{does not satisfy
      Condition (\ref{eq:wk})}\right\},\\
  \Ir_0^c = & \left\{ 1\leq k_1 < \ldots < k_d \leq d_{st}: \text{ for any
    } m_1, m_2 \in \{k_1,\ldots,k_d\} \text{ with } m_1\ldots m_2
  \right.\\
  & \left. G_{i_{m_1} j_{m_1} i_{m_2} j_{m_2}} \text{ satisfies
      Condition (\ref{eq:wk})}\right\},\\
\end{align*}
Obviously, $\Ir = \Ir_0 \bigcup \Ir_0^c$. For any subset $\Sr$ of
$\{k_1,\ldots, k_d\}$, we say that $\Sr$ satisfies (\ref{eq:lwk}) if 
\begin{equation}
  \label{eq:lwk}
  \text{For any } m_1\neq m_2\in \Sr,\ G_{i_{m_1} j_{m_1} i_{m_2}
    j_{m_2}} \text{ satisfies (\ref{eq:wk})}.
\end{equation}
For $2\leq l\leq d$, let 
\begin{align*}
  \Ir_{0l} = & \{1\leq k_1 <\ldots < k_d \leq d_{st}: \text{ the
    cardinality
    of the largest subset } \Sr \text{ is } l, \text{ where }\\
  & \Sr \text{ is a subset of } \{k_1,\ldots,k_d\} \text{ satisfies
    (\ref{eq:lwk})} \}\\
  \Ir_{01} = & \{1\leq k_1 <\ldots < k_d \leq d_{st}: \text{ For any }
  m_1,m_2 \in \{k_1,\ldots,k_d\} \text{ with } m_1 \neq m_2\\
  & G_{i_{m_1} j_{m_1} i_{m_2} j_{m_2}} \text{ does not satisfy (\ref{eq:wk})} \}\\
\end{align*}
Obviously, $\Ir_0^c = \Ir_{0d}$ and $\Ir_0 = \bigcup_{l=1}^{d_{st}-1}
\Ir_{0l}$. It is easy to show that $\Card(\Ir_{0l}) \leq
d_{st}^{l+\gamma(d-l)}$ and $\Card(\Ir_{0}^c) \leq {d_{st}
\choose d}$. It suffices to prove
\begin{align}
  \sum_{\Ir_0^c} \P (\lvert\bN\rvert_{\min}\geq y(d_{st},q_\alpha)^{1/2}\pm
  \epsilon_n(\log q_0)^{-1/2}) & =
  (1+o(1))\frac{1}{d!}\left\{\pi^{-1/2}\exp(-x/2)\right\}^d \label{eq:I0c}\\
  \sum_{\Ir_0}\P\left(\lvert \bN \rvert_{\min} \geq y(d_{st},q_\alpha)^{1/2}\pm
    \epsilon_n(\log q_0)^{-1/2} \right) & = o(1) \label{eq:I0}
\end{align}

We first prove (\ref{eq:I0}). 
Further divide $\Ir_{0l}$ as follows. Let $(k_1,\ldots,k_d)\in
\Ir_{0l}$ and let $\Sr_* \subseteq (k_1,\ldots,k_d)$ be the largest
cardinality subset satisfying (\ref{eq:lwk}). Define
\begin{align*}
  \Ir_{0l1} &  = \{(k_1,\ldots,k_d)\in \Ir_{0l}: \text{there exists an }
  a \not\in \Sr_* \text{ such that for some } b_1,b_2 \in \Sr_*\\
          & \phantom{aaaaa} \text{
  with }b_1\neq b_2, \text{ both } G_{i_a j_a i_{b_1}j_{b_1}} \text{ and
} G_{i_a j_a i_{b_2} j_{b_2}} \text{ is 3G1E  or  4G2E}. \}\\
  \Ir_{0l2} & = \Ir_{0l}\setminus \Ir_{0l1}.
\end{align*}

It is easy to see that $\Ir_{0l1} = \emptyset$ and $\Ir_{0l2} =
\Ir_{0l}$. Recall that $d$ is fixed and $l\leq d-1$. We can show that
$\Card(\Ir_{0l1}) \leq C_d d_{st}^{l-1 + \gamma(d-l+1)}$ and
$\Card(\Ir_{0l2})\leq C_d d_{st}^{l+\gamma(d-l)}$. Let $\Sr_* =
\{b_1,\ldots,b_l\}$ and $x(d_{st}) = y(d_{st},x)^{1/2}\pm\epsilon_n(\log d_{st})^{-1/2}$.

For any $(k_1,\ldots,k_d) \in \Ir_{0l}$, let $\vU_l$ be the covariance
matrix of $(N_{b_1},\ldots,N_{b_l})$. By (\ref{eq:wk}), $\lVert
\vU_l-\Ir_l\rVert_2 \leq O\{(\log q_0)^{-1-\alpha_0}\}$. Let $\lvert \by
\rvert_{\max} = \max_{1\leq i\leq l}\lvert y_i\rvert$ for
$\by=(y_1,\ldots,y_l)$. Then
\begin{align}
  \P\{\lvert \vN_d\rvert_{\min} \geq x(d_{st})\} & \leq \P\{\lvert
  N_{b_1}\rvert \geq x(d_{st}),\ldots,\lvert N_{b_l}\rvert \geq x(d_{st})\} \nonumber\\
  & = \frac{1}{(2\pi)^{l/2}\lvert\vU_l\rvert^{1/2}}\int_{\lvert \by
    \rvert_{\min}\geq x(d_{st})} \exp(-\frac{1}{2}\by^\trans\vU_l^{-1}\by)
  \ud
  \by \nonumber\\
  & = \frac{1}{(2\pi)^{l/2}\lvert\vU_l\rvert^{1/2}}\int_{\lvert \by
    \rvert_{\min}\geq x(d_{st}), \lvert \by \rvert_{\max} \leq (\log
    q_0)^{1/2+\alpha_0/4}} \exp(-\frac{1}{2}\by^\trans\vU_l^{-1}\by) \ud
  \by \nonumber\\
  & \phantom{aa} + O[\exp\{-(\log q_0)^{1+\alpha_0/2}/4\}]\nonumber\\
 & = \frac{1+O\{(\log q_0)^{-\alpha_0/2}\}}{(2\pi)^{l/2}} \int_{\lvert \by
    \rvert_{\min}\geq x(d_{st}), \lvert \by \rvert_{\max} \leq (\log
    q_0)^{1/2+\alpha_0/4}} \exp(-\frac{1}{2}\by^\trans\by) \ud
  \by\nonumber\\
  & \phantom{aa} + O[\exp\{-(\log q_0)^{1+\alpha_0/2}/4\}]\nonumber\\
  & = O(d_{st}^{-l}) \label{eq:integral}
\end{align}
Thus,
\[\sum_{I_{0l1}} \P (\lvert \vN\rvert_{\min} \geq x(d_{st})) \leq C_d
d_{st}^{-1+\gamma(d-l+1)}= o(1).\]
           
For $(k_1,\ldots,k_d) \in \Ir_{0l2}$, let $a_1 =\min\{a:\
a\in(k_1,\ldots,k_d), a\not\in \Sr_*\}$. WLOG, assume $G_{i_{a_1}
  j_{a_1} i_{b_1} j_{b_1}}$ is 3G1E or 4G2E.
Because $(k_1,\ldots,k_d) \in \Ir_{0l2}$, by definition of
$\Ir_{0l2}$, 
\begin{align*}
  \Cov(N_{a_1},N_{b_j}) &  = O((\log q_0)^{-1-\alpha_0}),\quad j = 2,
  \ldots, l\\
  \Cov(N_{b_i},N_{b_j}) & = O((\log q_0)^{-1-\alpha_0}),\quad i,j = 1,
  \ldots, l,\ i\neq j.
\end{align*}
Let $\vV_l$ be the covariance matrix of
$(N_{a_1},N_{b_1},\ldots,N_{b_l})$. It follows that
$\lVert \vV_l-\hat{\vV}_l\rVert_2 = O((\log q_0)^{-1-\alpha_0})$, where
$\hat{\vV}_l = \diag(\vD,\vI_{l-1})$ with $\vD$ to be the covariance
matrix of $(N_{a_1},N_{b_1})$.

By the conditions, for all $a_1$ and $b_1$,
\[
 \frac{ \lvert \E X_{i_{a_1}} Y_{j_{a_1}}  X_{i_{b_1}}
   Y_{j_{b_1}}\rvert}{ (E X_{i_{a_1}}^2 Y_{j_{a_1}}^2)^{1/2}(E
   X_{i_{b_1}}^2 Y_{j_{b_1}}^2)^{1/2}}  = \rho_{ss,
   i_{a_1}i_{b_1}}\rho_{tt,j_{a_1} j_{b_1}}\leq (\rho_0+1)/2.
\]

Using the similar argument as (\ref{eq:integral}), we can show that
\begin{align*}
  & \sum_{\Ir_{0l2}} \P\{\lvert N_{a_1} \rvert \geq x(d_{st}), \lvert
  N_{b_1}
  \rvert \geq x(d_{st}), \ldots, \lvert N_{b_l}\rvert \geq x(d_{st})\}\\
  \leq & C\sum_{\Ir_{0l2}}\left[ \P\{\lvert N_{a_1}\rvert \geq
    x(d_{st}), \lvert N_{b_1}\rvert \geq x(d_{st})\} \times
    d_{st}^{-(l-1)} + \exp\{-(\log
    q_0)^{1+\alpha_0/2}/4\} \right]\\
  \leq & C\sum_{\Ir_{0l2}}\left[ d_{st}^{-1-(1-\rho_0)/(3+\rho_0)}
    \times d_{st}^{-(l-1)}
    + \exp(-(\log q_0)^{1+\alpha_0/2}/4) \right]\\
  \leq & C d_{st}^{-\frac{1-\rho_0}{3+\rho_0}+\gamma(d-l)} + q_0^{-M} = o(1)\\
\end{align*}
Thus (\ref{eq:I0}) is proved. Following the same argument as
(\ref{eq:integral}) and $\Card(I_{0}^c) = (1+o(1)){d_{st}\choose d}$,
we can prove (\ref{eq:I0c}).
\end{proof}

\section{Simulated Network in Section 5.2}

\begin{figure}[htbp] 
   \centering
   \includegraphics[width=6in]{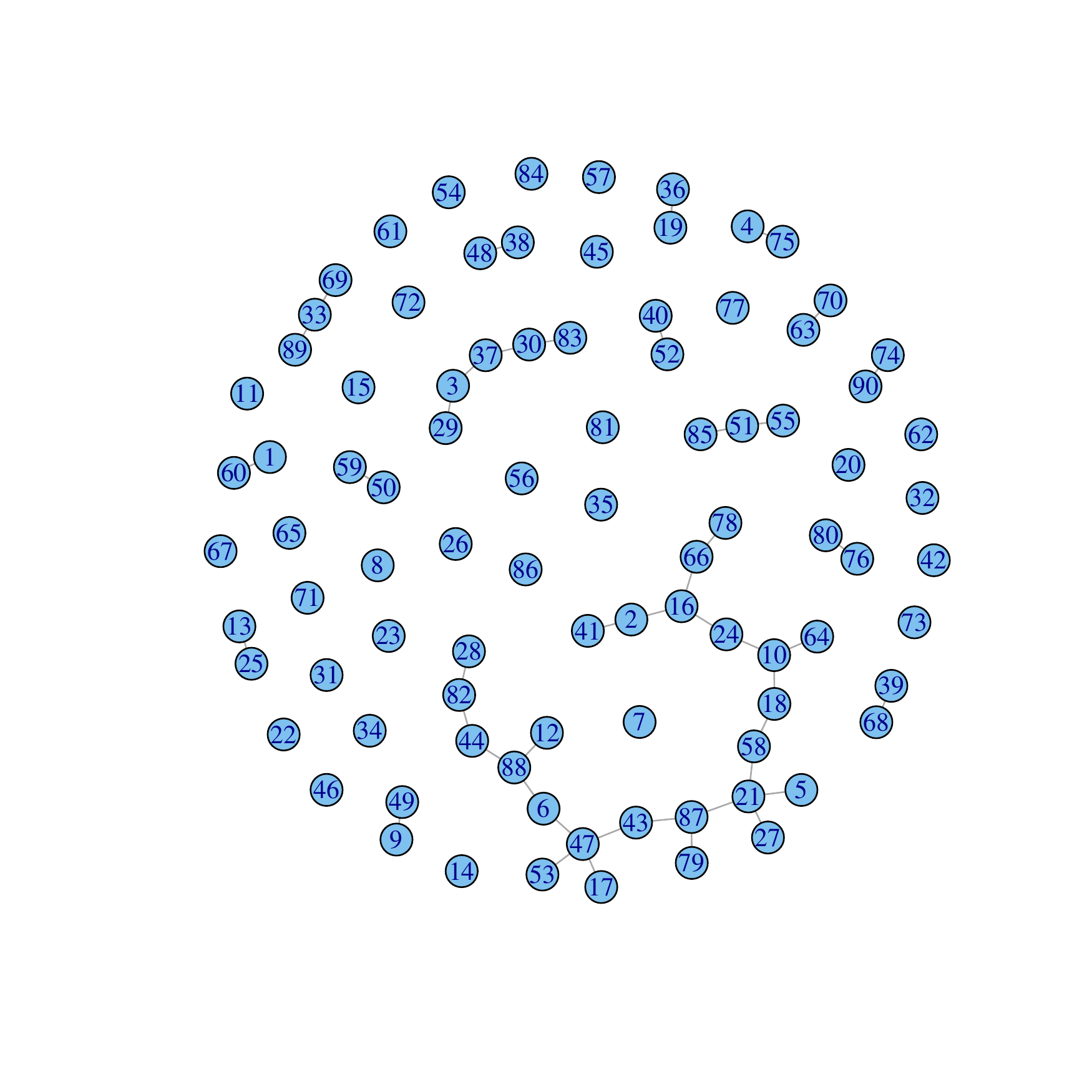} 
   \caption{Simulated network on 90 regions using the
     Er\"{d}os-R\'{e}nyi model}
   \label{fig:simulNet}
\end{figure}

\end{document}